\documentclass[11pt]{article}

\usepackage[table]{xcolor}

\usepackage{acl}

\usepackage{times}
\usepackage{latexsym}

\usepackage[T1]{fontenc}

\usepackage[utf8]{inputenc}

\usepackage{microtype}

\usepackage{inconsolata}

\usepackage{graphicx}

\usepackage{dsfont}
\usepackage{multirow}
\usepackage{enumitem}
\usepackage{wrapfig}
\usepackage{subcaption}
\usepackage{amsthm}
\usepackage{amsmath}
\usepackage{amssymb}
\usepackage[american]{babel}
\usepackage{setspace}
\usepackage{cleveref}
\usepackage{caption}
\usepackage{threeparttable}
\usepackage{comment}

\newtheorem{lemma}{Lemma}

\usepackage{tcolorbox}
\usepackage{setspace}
\usepackage{adjustbox}
\usepackage{bbding}
\usepackage{float}
\usepackage{placeins}
\usepackage[ruled]{algorithm2e} %
\usepackage{wrapfig}
\usepackage[normalem]{ulem}
\usepackage{bm}
\usepackage{mathtools}
\usepackage{listings}
\usepackage{tikz} 
\usepackage{booktabs}

\newtheorem{proposition}{Proposition}

\definecolor{humanblue}{RGB}{74, 152, 237} 
\definecolor{llmred}{RGB}{240, 134, 140}  

\newcommand{\posnum}[1]{\cellcolor{humanblue!30}{#1}}
\newcommand{\negnum}[1]{\cellcolor{llmred!30}{#1}}

\newcommand{\avgshadetrain}[2]{%
  \begingroup
  \pgfmathsetmacro{\ratio}{abs(#1)/abs(#2)}%
  \pgfmathsetmacro{\ratio}{min(\ratio,1)}%
  \pgfmathsetmacro{\gamma}{1.5}%
  \pgfmathsetmacro{\intensity}{100*pow(\ratio,\gamma)}%
  \tikz[baseline=(n.base)]{%
    \node[
      inner sep=2.2pt,            %
      outer sep=0pt,
      rounded corners=1pt,        %
      fill=llmred!\intensity!white,
      text height=1.8ex,          %
      text depth=0pt              %
    ] (n) {\strut #1};
  }%
  \endgroup
}

\newcommand{\avgshadeinfer}[2]{%
  \begingroup
  \pgfmathsetmacro{\ratio}{abs(#1)/abs(#2)}%
  \pgfmathsetmacro{\ratio}{min(\ratio,1)}%
  \pgfmathsetmacro{\gamma}{2}%
  \pgfmathsetmacro{\intensity}{100*pow(\ratio,\gamma)}%
  \tikz[baseline=(n.base)]{%
    \node[
      inner sep=2.2pt,            %
      outer sep=0pt,
      rounded corners=1pt,        %
      fill=llmred!\intensity!white,
      text height=1.8ex,          %
      text depth=0pt              %
    ] (n) {\strut #1};
  }%
  \endgroup
}

\newcommand{\traincolorbarvertical}[1][2]{%
  \begin{tikzpicture}[x=1cm,y=1cm]
    \def\H{4}   %
    \def\W{0.35}  %
    \pgfmathsetmacro{\G}{#1}

    \foreach \i in {0,...,119} {%
      \pgfmathsetmacro{\fLow}{\i/120}%
      \pgfmathsetmacro{\fHigh}{(\i+1)/120}%
      \pgfmathsetmacro{\fMid}{(\fLow+\fHigh)/2}%
      \pgfmathsetmacro{\intensity}{100*pow(\fMid,\G)}%

      \pgfmathsetmacro{\yLow}{\H*(1-\fHigh)}%
      \pgfmathsetmacro{\yHigh}{\H*(1-\fLow)}%

      \path[draw=none, fill=llmred!\intensity!white]
        (0,\yLow) rectangle (\W,\yHigh);
    }

    \foreach \p in {0,10,20,30,40,50,60,70,80,90,100} {%
      \pgfmathsetmacro{\frac}{\p/100}%
      \pgfmathsetmacro{\y}{\H*(1-\frac)}

      \draw[white, line width=0.35pt] (0,\y) -- (\W,\y);

      \node[right] at (\W+0.06,\y) {\footnotesize \p};
    }

    \node[right] at (\W,\H+0.5) {\small \textbf{\%}};
  \end{tikzpicture}%
}

\newcommand{\debiascolorbarvertical}[1][2]{%
  \begin{tikzpicture}[x=1cm,y=1cm]
    \def\H{2.5}   %
    \def\W{0.32}  %
    \pgfmathsetmacro{\G}{#1}

    \foreach \i in {0,...,99} {%
      \pgfmathsetmacro{\fLow}{\i/100}%
      \pgfmathsetmacro{\fHigh}{(\i+1)/100}%
      \pgfmathsetmacro{\fMid}{(\fLow+\fHigh)/2}%
      \pgfmathsetmacro{\intensity}{100*pow(\fMid,\G)}%
      \pgfmathsetmacro{\yLow}{\H * (1-\fHigh)}%
      \pgfmathsetmacro{\yHigh}{\H * (1-\fLow)}%
      \path[draw=none, fill=llmred!\intensity!white]
        (0,\yLow) rectangle (\W,\yHigh);
    }

    \foreach \p in {0,20,40,60,80,100} {%
      \pgfmathsetmacro{\frac}{\p/100}%
      \pgfmathsetmacro{\y}{\H * (1-\frac)}%
      \draw[white, line width=0.35pt] (0,\y) -- (\W,\y);
      \node[right] at (\W+0.06,\y) {\footnotesize \p};
    }

    \node[right] at (\W,\H+0.5) {\small \textbf{\%}};
  \end{tikzpicture}%
}

\title{Relevance Supervision Can Teach Neural Retrievers to Prefer LLM-Generated Text}

\author{
  \textbf{Wei Huang\textsuperscript{1,2}},
  \textbf{Keping Bi\textsuperscript{1,2}\thanks{Corresponding author.}},
  \textbf{Yinqiong Cai\textsuperscript{3}},
  \textbf{Wei Chen\textsuperscript{1,2}},
\\
  \textbf{Jiafeng Guo\textsuperscript{1,2}},
  \textbf{Xueqi Cheng\textsuperscript{1,2}}
\\
  \textsuperscript{1}State Key Laboratory of AI Safety,
\\
  Institute of Computing Technology, Chinese Academy of Sciences
\\
  \textsuperscript{2}University of Chinese Academy of Sciences,
  Beijing, China
\\
  \textsuperscript{3}Baidu Inc., Beijing, China
\\
    \texttt{\{huangwei21b,bikeping,chenwei2022,guojiafeng,cxq\}@ict.ac.cn}
    \\
    \texttt{caiyinqiong@baidu.com}
}

\begin{document}
\maketitle
\begin{abstract}

Recent studies show that neural retrievers often exhibit source bias, ranking LLM-generated passages above semantically similar human-written ones, which threatens the reliability of modern information access systems. While this bias is often viewed as a model-level property, we show that relevance supervision itself can induce it. Empirically, source bias is clearest in relevance-supervised retrievers, and the same encoders become more LLM-preferring after relevance fine-tuning. We trace this effect to surface-level artifacts in supervised retrieval data: annotated positives and negatives differ systematically in properties such as fluency and lexical specificity, and these differences parallel those between LLM-generated and human-written passages. As a result, retrievers can learn artifact-based shortcuts for relevance. We further show theoretically that, under contrastive learning, artifacts that reliably separate positives from negatives can be encoded into learned relevance scores. Motivated by this mechanism, we propose complementary strategies that reduce source bias during training and at inference: balancing artifact distributions in negative sampling and removing the artifact-aligned embedding component, respectively. Overall, our findings support a data-centric view of source bias: retrievers may prefer LLM-generated text not merely because of inherent model properties, but because relevance supervision teaches them to rely on surface-level artifacts.\begingroup
\renewcommand{\thefootnote}{\fnsymbol{footnote}}%
\footnote[2]{The code for this work is publicly available at \url{https://github.com/Trustworthy-Information-Access/Relevance-Supervision-Source-Bias}.}%
\endgroup
\end{abstract}

\section{Introduction}
\label{sec:intro}

The rapid rise of large language models (LLMs) has reshaped the information landscape, creating corpora in which human-written and LLM-generated texts increasingly coexist. In this setting, recent work shows that neural retrievers can rank LLM-generated passages above semantically similar human-written ones, a phenomenon known as \emph{source bias}~\cite{daiBias2024,daiNeural2024a}. This preference matters because retrievers serve as front ends for search, recommendation, and retrieval-augmented generation: if they systematically favor generated text, they may surface fluent but less reliable passages, reduce the visibility of human-authored content, and amplify generated text through feedback loops~\cite{chenSpiral2024a,zhouExploring2025b}.

Prior work has mainly explained source bias through model-side or text-level factors, such as architectural similarities between PLM-based retrievers and LLMs~\cite{daiBias2024,daiNeural2024a} or retrievers' preference for low-perplexity passages~\cite{wangPerplexity2024a}. Accordingly, existing mitigation methods intervene either on the retriever side~\cite{wangPerplexity2024a} or on the generation side~\cite{daiMitigating2025}. Yet these explanations leave open a data-centric possibility: relevance supervision itself may induce source bias. In supervised retrieval data, relevant and non-relevant passages may differ beyond semantic relevance. They may also differ in surface-level \emph{artifacts}, such as fluency or lexical specificity, that retrievers can exploit as shortcuts.

In this paper, we examine whether such artifact-based shortcuts make supervised retrievers favor LLM-generated text. Specifically, we ask whether artifacts that separate positives from negatives during training can transfer to source preferences at test time. This leads to three research questions:

\textbf{RQ1: Where does source bias emerge?} We first examine source preference across three retriever families: relevance-supervised retrievers, general-purpose embedding models, and unsupervised retrievers. We find that preference for LLM-generated passages is strongest and most consistent in relevance-supervised retrievers, and that fine-tuning unsupervised retrievers on MS~MARCO shifts them toward LLM-generated passages.

\textbf{RQ2: How can relevance supervision shape source bias?} Across retrieval datasets, positive and negative documents differ systematically in surface-level properties such as fluency and lexical specificity, and these patterns align with those between LLM-generated and human-written text. We provide a theoretical analysis showing how contrastive objectives can encode such artifact signals when they help separate positives from negatives.

\textbf{RQ3: How can source bias be mitigated?} Motivated by this mechanism, we test two strategies: balancing artifact distributions during negative sampling at training time and suppressing the artifact-aligned component of retriever embeddings at inference time. Both reduce source bias.

Overall, our findings support a data-centric view of source bias: rather than being an inherent property of neural retrievers, source bias can emerge when relevance supervision encourages models to rely on artifact-based shortcuts that also distinguish LLM-generated text from human-written text. This perspective suggests complementary mitigation directions at both training and inference time, including artifact-aware data construction and embedding-level debiasing.

\section{Related Work}
\label{sec:related_work}

\paragraph{Source Bias in Information Retrieval.}
Recent work has shown that neural retrievers can prefer LLM-generated passages over semantically similar human-written ones, a phenomenon known as \emph{source bias}~\cite{daiNeural2024a,daiBias2024}. 
Cocktail~\cite{daiCocktail2024a} extends BEIR~\cite{thakurBEIR2021a} with LLM-generated counterparts, enabling systematic evaluation across retrieval datasets. 
Related source-preference effects have also been observed in multimodal retrieval, recommendation, and retrieval-augmented generation~\cite{xuInvisible2024a,zhouExploring2025b,chenSpiral2024a}.

\paragraph{Mechanisms and Mitigation.}
Prior work links source bias to architectural similarities between PLM-based retrievers and LLMs~\cite{daiNeural2024a}, retrievers' preference for low-perplexity documents~\cite{wangPerplexity2024a}, or distribution mismatch~\cite{daiBias2024}. These explanations motivate mitigation strategies: retriever-side debiasing via CDC~\cite{wangPerplexity2024a} and generation-side alignment of LLM passages~\cite{daiMitigating2025}. These methods mitigate source bias but do not explain how retrievers acquire this preference during training. We instead study a data-centric mechanism: relevance supervision makes surface artifacts predictive of relevance, leading retrievers to learn shortcuts that favor LLM-generated text.
\section{RQ1: Where Does Source Bias Emerge?}
\label{sec:rq1}

Prior work reports that neural retrievers often rank LLM-generated passages above semantically similar human-written ones~\cite{daiBias2024,daiNeural2024a}. However, prior studies mainly evaluate retrievers trained with large-scale relevance supervision, such as MS~MARCO~\cite{nguyen2016msmarco}. This raises the question of whether source bias appears broadly across neural retrievers, or is especially associated with relevance-supervised retrieval.

We address this question with two analyses. First, we compare three retrieval model families: relevance-supervised retrievers, general-purpose embedding models, and unsupervised retrievers. Second, we fine-tune several unsupervised retrievers on MS~MARCO to test whether relevance supervision makes the same base encoders favor LLM-generated passages.

\begin{table*}[t]
\centering
\caption{$\Delta$NDSR@5 across 14 datasets and 13 retrievers spanning three model families. 
\textcolor{llmred!80!black}{Negative} values indicate LLM preference and 
\textcolor{humanblue!80!black}{positive} values indicate human preference. 
$*$ denotes $p<0.05$.}
\label{tab:rq1_main}
\renewcommand{\arraystretch}{1.0}
\setlength{\tabcolsep}{3.5pt}
\resizebox{1.0\textwidth}{!}{
\begin{tabular}{lccccc | ccccc | ccc}
\toprule
\multirow{2}{*}{Dataset (↓)} & \multicolumn{5}{c}{Relevance-Supervised Retrievers} & \multicolumn{5}{c}{General-Purpose Embedding Models} & \multicolumn{3}{c}{Unsupervised Retrievers} \\
\cmidrule(lr){2-6} \cmidrule(lr){7-11} \cmidrule(lr){12-14}
& ANCE & TAS-B & coCondenser & RetroMAE & DRAGON & BGE & BCE & GTE & E5 & M3E & Contriever & E5-unsup & SimCSE \\
\midrule
MS~MARCO       & \negnum{-0.042*} & \negnum{-0.121*} & \negnum{-0.020*} & \negnum{-0.083*} & \negnum{-0.083*} & \negnum{-0.021*} & \posnum{0.084*} & \negnum{-0.074*} & \negnum{-0.036*} & \posnum{0.053*} & \posnum{0.280*} & \posnum{0.094*} & \posnum{0.384*} \\
DL19           & \negnum{-0.073} & \negnum{-0.224*} & \negnum{-0.072} & \negnum{-0.186*} & \negnum{-0.233*} & \negnum{-0.017} & \posnum{0.119} & \negnum{-0.178*} & \posnum{0.015} & \posnum{0.139} & \posnum{0.271*} & \posnum{0.086} & \posnum{0.428*} \\
DL20           & \negnum{-0.034} & \negnum{-0.072} & \negnum{-0.079} & \negnum{-0.088} & \negnum{-0.121*} & \posnum{0.057} & \posnum{0.048} & \negnum{-0.049} & \posnum{0.012} & \posnum{0.203*} & \posnum{0.275*} & \posnum{0.190*} & \posnum{0.389*} \\
NQ             & \negnum{-0.042*} & \negnum{-0.078*} & \negnum{-0.072*} & \negnum{-0.060*} & \negnum{-0.099*} & \negnum{-0.078*} & \posnum{0.324*} & \negnum{-0.003} & \posnum{0.153*} & \posnum{0.040*} & \posnum{0.186*} & \posnum{0.228*} & \posnum{0.140*} \\
NFCorpus       & \negnum{-0.087*} & \negnum{-0.082*} & \negnum{-0.068*} & \negnum{-0.081*} & \negnum{-0.079*} & \posnum{0.030} & \negnum{-0.064*} & \negnum{-0.142*} & \posnum{0.034} & \negnum{-0.143*} & \negnum{-0.083*} & \negnum{-0.348*} & \posnum{0.127*} \\
TREC-COVID     & \negnum{-0.162*} & \negnum{-0.328*} & \negnum{-0.340*} & \negnum{-0.194*} & \negnum{-0.134*} & \posnum{0.014} & \negnum{-0.025} & \negnum{-0.236*} & \negnum{-0.118} & \negnum{-0.085} & \negnum{-0.135*} & \negnum{-0.224*} & \posnum{0.162*} \\
HotpotQA       & \negnum{-0.020*} & \negnum{-0.018*} & \negnum{-0.014*} & \negnum{-0.019*} & \negnum{-0.018*} & \posnum{0.061*} & \posnum{0.184*} & \posnum{0.010*} & \posnum{0.078*} & \posnum{0.063*} & \negnum{-0.273*} & \negnum{-0.091*} & \posnum{0.097*} \\
FiQA-2018      & \negnum{-0.179*} & \negnum{-0.170*} & \negnum{-0.219*} & \negnum{-0.205*} & \negnum{-0.161*} & \negnum{-0.150*} & \posnum{0.414*} & \negnum{-0.050*} & \negnum{-0.116*} & \posnum{0.102*} & \negnum{-0.068*} & \negnum{-0.052*} & \posnum{0.210*} \\
Touché-2020    & \negnum{-0.168} & \negnum{-0.247*} & \negnum{-0.226*} & \negnum{-0.175} & \negnum{-0.178*} & \negnum{-0.042} & \posnum{0.218*} & \negnum{-0.017} & \negnum{-0.185*} & \posnum{0.242*} & \negnum{-0.133*} & \negnum{-0.073} & \posnum{0.027} \\
DBpedia        & \negnum{-0.097*} & \negnum{-0.042*} & \negnum{-0.054*} & \negnum{-0.059*} & \negnum{-0.057*} & \posnum{0.017} & \posnum{0.069*} & \negnum{-0.035*} & \posnum{0.003} & \posnum{0.019} & \negnum{-0.130*} & \negnum{-0.062*} & \posnum{0.064*} \\
SCIDOCS        & \negnum{-0.040*} & \negnum{-0.054*} & \negnum{-0.058*} & \negnum{-0.073*} & \negnum{-0.048*} & \negnum{-0.061*} & \posnum{0.517*} & \negnum{-0.046*} & \posnum{0.010} & \posnum{0.275*} & \posnum{0.028*} & \posnum{0.059*} & \posnum{0.268*} \\
FEVER          & \negnum{-0.200*} & \negnum{-0.029*} & \negnum{-0.037*} & \negnum{-0.010*} & \negnum{-0.043*} & \posnum{0.040*} & \posnum{0.306*} & \negnum{-0.027*} & \posnum{0.031*} & \posnum{0.031*} & \posnum{0.028*} & \negnum{-0.008*} & \posnum{0.031*} \\
Climate-FEVER  & \negnum{-0.314*} & \negnum{-0.083*} & \negnum{-0.153*} & \negnum{-0.105*} & \negnum{-0.091*} & \negnum{-0.038*} & \posnum{0.642*} & \negnum{-0.080*} & \posnum{0.215*} & \posnum{0.123*} & \negnum{-0.003} & \posnum{0.017} & \posnum{0.070*} \\
SciFact        & \negnum{-0.025} & \negnum{-0.058*} & \negnum{-0.049*} & \negnum{-0.048*} & \negnum{-0.041*} & \posnum{0.011} & \posnum{0.015} & \negnum{-0.079*} & \posnum{0.004} & \negnum{-0.206*} & \posnum{0.017} & \negnum{-0.101*} & \negnum{-0.059*} \\
\bottomrule
\end{tabular}
}
\end{table*}

\subsection{Experimental Setup}
\label{subsec:rq1_design}

We evaluate 13 neural retrievers from three families. 
The relevance-supervised group includes ANCE~\cite{xiongApproximate2020}, TAS-B~\cite{hofstatterEfficiently2021}, coCondenser~\cite{gaoUnsupervised2022}, RetroMAE~\cite{xiaoRetroMAE2022}, and DRAGON~\cite{linHow2023}. 
The general-purpose embedding group includes BGE~\cite{xiaoCPack2024}, BCE~\cite{youdao_bcembedding_2023}, GTE~\cite{liGeneral2023}, E5~\cite{wangText2024}, and M3E~\cite{Moka_Massive_Mixed_Embedding}. 
The unsupervised group includes Contriever~\cite{izacardUnsupervised2022}, E5-unsup~\cite{wangText2024}, and unsupervised SimCSE~\cite{gaoSimCSE2021}. 
We treat these groups as common training regimes rather than a strict supervised/unsupervised partition, since some general-purpose embedding models may include retrieval-style or weakly supervised data in their training mixtures.

We evaluate on 14 Cocktail datasets~\cite{daiCocktail2024a}, which augment BEIR collections~\cite{thakurBEIR2021a} with LLM-generated counterparts designed to preserve the semantics of human-written passages. The original study reports high semantic consistency on a manually checked subset. Dataset statistics, model checkpoints, and resource links are provided in Appendix~\ref{app:dataset_stats} and Appendix~\ref{app:resources}.

Following prior work~\cite{huangHow2025}, we measure source preference by the top-ranked source composition, using $\Delta$NDSR@5:
\begin{align}
\mathrm{NDSR}_{c}@k
&=
\frac{\sum_{i=1}^k 
\mathds{1}\!\left(\mathrm{source}(d_i)=c\right)\, w_i}
{\sum_{i=1}^k w_i}, \nonumber
\\
\Delta \mathrm{NDSR}@k 
&=
\mathrm{NDSR}_{\text{Human}}@k
-
\mathrm{NDSR}_{\text{LLM}}@k, \nonumber
\end{align}
where $w_i=1/\log_2(1+i)$. $\Delta\mathrm{NDSR}@k$ ranges from $-1$ to $1$: positive values indicate preference for human-written passages, while negative values indicate preference for LLM-generated passages.

\subsection{Results}
\label{subsec:rq1_results}

\begin{table}[t]
\centering
\caption{$\Delta$NDSR@5 for the unsupervised retrievers in Table~\ref{tab:rq1_main} after MS~MARCO fine-tuning. The ``-FT'' suffix denotes fine-tuning on MS~MARCO.
\textcolor{llmred!80!black}{Negative} values indicate LLM preference and
\textcolor{humanblue!80!black}{positive} values indicate human preference.
$*$ denotes $p<0.05$.}
\label{tab:rq1_ctrl}
\renewcommand{\arraystretch}{1.0}
\setlength{\tabcolsep}{3pt}
\resizebox{\columnwidth}{!}{%
\begin{tabular}{l ccc}
\toprule
\multirow{2}{*}{Dataset (↓)} & \multicolumn{3}{c}{MS~MARCO Fine-Tuned Retrievers} \\
\cmidrule(lr){2-4}
& Contriever-FT & E5-FT & SimCSE-FT \\
\midrule
MS~MARCO      &  \posnum{0.012*} & \negnum{-0.044*} & \negnum{-0.053*} \\
DL19          &  \negnum{-0.035} & \negnum{-0.198*} & \negnum{-0.133} \\
DL20          &  \posnum{0.121*} &  \posnum{0.022} & \negnum{-0.178*} \\
NQ            &  \negnum{-0.038*} & \negnum{-0.051*} & \negnum{-0.060*} \\
NFCorpus      &  \negnum{-0.139*} & \negnum{-0.189*} & \negnum{-0.060*} \\
TREC-COVID    &  \negnum{-0.282*} & \negnum{-0.271*} & \negnum{-0.205*} \\
HotpotQA      &  \negnum{-0.004} & \negnum{-0.019*} & \negnum{-0.013*} \\
FiQA-2018     &  \negnum{-0.215*} & \negnum{-0.212*} & \negnum{-0.189*} \\
Touché-2020   &  \negnum{-0.087*} & \negnum{-0.196*} & \negnum{-0.169*} \\
DBpedia       &  \negnum{-0.010} & \negnum{-0.036*} & \negnum{-0.053*} \\
SCIDOCS       &  \negnum{-0.050*} & \negnum{-0.072*} & \negnum{-0.041*} \\
FEVER         &  \negnum{-0.018*} & \negnum{-0.064*} & \posnum{0.000} \\
Climate-FEVER &  \negnum{-0.099*} & \negnum{-0.091*} & \negnum{-0.049*} \\
SciFact       &  \negnum{-0.086*} & \negnum{-0.077*} & \negnum{-0.044*} \\
\bottomrule
\end{tabular}%
}
\end{table}
\paragraph{Source preference across retriever families.}

\begin{figure*}[t]
  \centering
  \begin{subfigure}[t]{0.48\linewidth}
    \centering
    \includegraphics[width=0.95\linewidth]{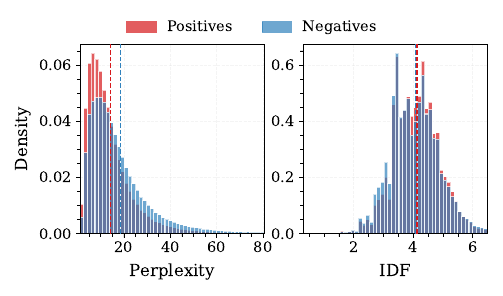}
    \caption{Positives vs. Negatives.}
    \label{fig:ppl_idf_supervision}
  \end{subfigure}
  \hfill
  \begin{subfigure}[t]{0.48\linewidth}
    \centering
    \includegraphics[width=0.95\linewidth]{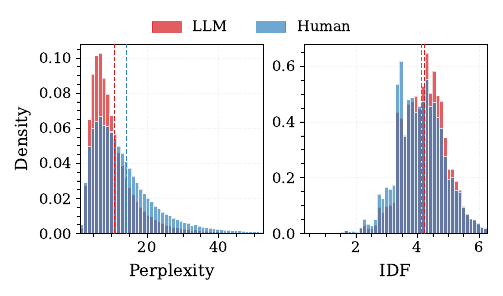}
    \caption{LLM-generated vs. Human-written.}
    \label{fig:ppl_idf_source}
  \end{subfigure}

  \caption{
    PPL and IDF distributions for (a) annotated positives versus training negatives and (b) LLM-generated versus human-written passages. Dashed lines indicate means. In both comparisons, positives and LLM-generated passages exhibit lower PPL and higher IDF, indicating parallel artifact-related shifts.
  }
  \label{fig:ppl_idf}
\end{figure*}

Table~\ref{tab:rq1_main} reveals a clear separation by training regime across 14 datasets and 13 retrievers. Relevance-supervised retrievers favor LLM-generated passages in nearly every model--dataset combination, whereas general-purpose embedding models and unsupervised retrievers show mixed source preferences. This cross-family contrast associates systematic LLM preference with relevance-supervised training rather than neural retrieval in general. Retrieval effectiveness is reported in Appendix~\ref{app:rq1_effectiveness}.

\paragraph{Effect of relevance fine-tuning.}
To isolate relevance supervision, we hold each base encoder family fixed and introduce relevance-labeled query--passage training under a common MS~MARCO fine-tuning setup. As shown in Table~\ref{tab:rq1_ctrl}, Contriever, E5-unsup, and SimCSE all shift consistently toward LLM-generated passages. For example, Contriever moves from $+0.280$ to $+0.012$ on MS~MARCO, a $0.268$ shift toward LLM-generated text. This pattern is stable across three random seeds: the cross-seed standard deviations of mean absolute $\Delta$NDSR@5 range from only $0.003$ to $0.008$ (Appendix Table~\ref{tab:rq1_multiseed}). With encoder families held fixed, these results directly attribute the shift in source preference to relevance fine-tuning.

\paragraph{Robustness to passage length.}
We assess whether the source preference is robust to passage length. 
LLM-generated passages in Cocktail are often shorter than their human-written counterparts, and neural retrievers can be sensitive to length-related artifacts~\cite{thakur2024systematic,fayyaz-etal-2025-collapse}. We therefore construct length-controlled variants by lengthening LLM-generated passages while preserving semantics. As detailed in Appendix~\ref{app:length_effects}, relevance-supervised retrievers continue to prefer LLM-generated passages after their length advantage is reversed. Passage length modulates the magnitude of source preference but does not account for its direction.

\paragraph{Summary.}
Across the cross-family comparison, controlled fine-tuning, and length-controlled evaluation, the evidence converges: relevance supervision systematically reorients neural retrievers toward LLM-generated text. This identifies relevance supervision as a key driver of source bias and motivates our next question: how does relevance supervision produce this preference?

\section{RQ2: How Can Relevance Supervision Shape Source Bias?}
\label{sec:rq2}

Having observed that source bias is most systematic in relevance-supervised retrieval pipelines, we now ask how relevance supervision may encourage this behavior. We hypothesize that positive and negative passages in supervised retrieval data can differ in \emph{surface-level artifacts} that correlate with relevance labels but are not identical to query-specific semantic relevance.
Here, an \emph{artifact} is a largely query-independent passage property that systematically separates positives from negatives, creating a corpus-level relevance shortcut. Examples include fluency and lexical specificity. When these properties also characterize LLM-generated text, retrievers can learn to assign it higher scores.

We examine this mechanism in three steps: first using illustrative linguistic diagnostics, then drawing on embedding-space evidence, and finally adopting a theoretical perspective.

\subsection{Illustrative Linguistic Diagnostics}
\label{subsec:rq2_linguistic}

We first examine two simple linguistic diagnostics that may reflect surface-level differences: perplexity (PPL) and inverse document frequency (IDF). PPL measures how predictable a passage is under a language model and can reflect fluency or textual regularity, while IDF reflects lexical specificity. Together, they provide an interpretable, low-dimensional view of possible artifact differences associated with relevance labels and source type.

For a passage $d=(w_1,\ldots,w_{|d|})$, we compute
\[
\mathrm{PPL}(d)=
\exp\left(
-\frac{1}{|d|}
\sum_{i=1}^{|d|}
\log p_\theta(w_i\mid w_{<i})
\right),
\]
using Llama-3-8B-Instruct~\citep{dubey2024llama}.
For IDF, we use
$
\mathrm{IDF}(t)=\log \frac{N}{1+\mathrm{df}(t)}
$
and take the median token IDF as the passage-level score, with statistics estimated on the full MS~MARCO collection using the Apache Lucene tokenizer~\cite{hatcher2004lucene}.

\paragraph{Positives vs. Negatives.}
We treat MS~MARCO passages annotated as relevant to at least one training query as positives. We compare these positives with the remaining collection, which serves as the candidate pool for sampled or mined training negatives. This choice reflects the fact that the exact negatives used by released retrievers are often unavailable and may vary across training pipelines. We refer to this pool as the ``negative pool'' for brevity, while noting that it may contain unannotated positives and is not a clean set of judged non-relevant passages. Thus, the comparison should be interpreted as a corpus-level diagnostic of surface-level differences between positives and likely negative candidates, rather than as a reconstruction of any particular retriever's training negatives.

As shown in Figure~\ref{fig:ppl_idf_supervision}, positives exhibit lower PPL and slightly higher IDF than the negative pool. 
The PPL difference is more pronounced, whereas the IDF difference is smaller but directionally consistent. 
Together, these patterns indicate a distributional tendency for positives to be more predictable under the language model and marginally more lexically specific. 
Appendix~\ref{app:linguistic_appendix} provides detailed statistics and additional evidence that similar PPL imbalances appear across other IR datasets, suggesting that this pattern is not unique to MS~MARCO.

\paragraph{LLM-generated vs. Human-written Passages.}
We next compare LLM-generated passages with their human-written counterparts across the 14 BEIR-derived datasets in Cocktail. Figure~\ref{fig:ppl_idf_source} shows representative results on MS~MARCO, where LLM-generated passages exhibit lower PPL and higher IDF. The same pattern therefore appears in both comparisons: positive passages relative to negatives and LLM-generated passages relative to human-written ones both have lower PPL and higher IDF. Full results are provided in Appendix~\ref{app:linguistic_appendix}. Using Qwen3-8B~\citep{qwen3technicalreport} as an alternative PPL estimator confirms that this contrast is robust to language-model choice: LLM-generated passages receive lower PPL across all 12 unique corpora (Appendix Table~\ref{tab:qwen_ppl}).

\paragraph{Artifact--Score Association.}
To determine whether these corpus-level differences are associated with retrieval behavior, we analyze approximately 57.6 million query--document pairs. We z-normalize document scores within each query and correlate them with document-level PPL and IDF. Across the five relevance-supervised retrievers, PPL is weakly negatively correlated with retrieval scores ($-0.033$ to $-0.021$), whereas IDF is weakly positively correlated with them ($0.003$ to $0.012$); all are statistically significant (Table~\ref{tab:artifact_score_corr}). Thus, lower-PPL and higher-IDF passages tend to receive higher scores, showing that these artifacts are reflected in retrieval rankings.

\begin{table}[t]
\centering
\small
\caption{Pearson correlations between query-normalized document scores and linguistic artifacts over approximately 57.6 million query--document pairs.}
\label{tab:artifact_score_corr}
\renewcommand{\arraystretch}{1.0}
\setlength{\tabcolsep}{5pt}
\begin{tabular}{lrr}
\toprule
Retriever & PPL vs.\ score & IDF vs.\ score \\
\midrule
ANCE        & -0.033 & 0.005 \\
TAS-B       & -0.022 & 0.006 \\
coCondenser & -0.021 & 0.007 \\
RetroMAE    & -0.030 & 0.003 \\
DRAGON      & -0.023 & 0.012 \\
\bottomrule
\end{tabular}
\end{table}

\paragraph{Summary.}
Together, these diagnostics establish three converging patterns: positives and LLM-generated passages share lower PPL and higher IDF; the PPL contrast persists across language-model families; and both artifacts are systematically associated with retrieval scores. Because PPL and IDF capture only selected aspects of textual variation, we next ask whether retriever embeddings encode a corresponding direction that aligns the positive--negative supervision contrast with the LLM--human source contrast.

\subsection{Embedding-Space Shifts}
\label{subsec:rq2_embedding}

\begin{figure*}[t]
  \centering
  \begin{tabular}{
    @{}
    b{0.46\linewidth}
    @{\hspace{0.05\linewidth}}
    b{0.46\linewidth} 
    @{} 
  }
    \setcounter{subfigure}{0} 
    \begin{subfigure}[t]{\linewidth}
      \centering
      \includegraphics[width=0.85\linewidth]{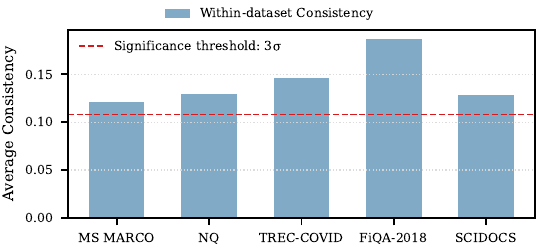}
      \caption{Within-dataset consistency.}
      \label{fig:rq2-a}
    \end{subfigure}

    \setcounter{subfigure}{2} 
    \begin{subfigure}[t]{\linewidth}
      \centering
      \includegraphics[width=0.85\linewidth]{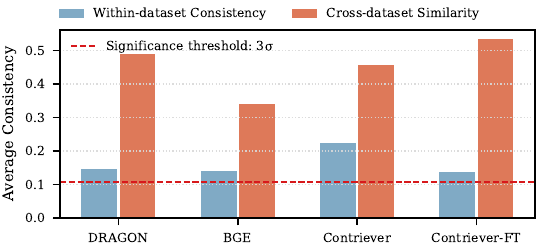}
      \caption{Cross-retriever consistency.}
      \label{fig:rq2-c}
    \end{subfigure}
    & 
    \setcounter{subfigure}{1} 
    \begin{subfigure}[t]{\linewidth}
      \centering
      \includegraphics[width=\linewidth]{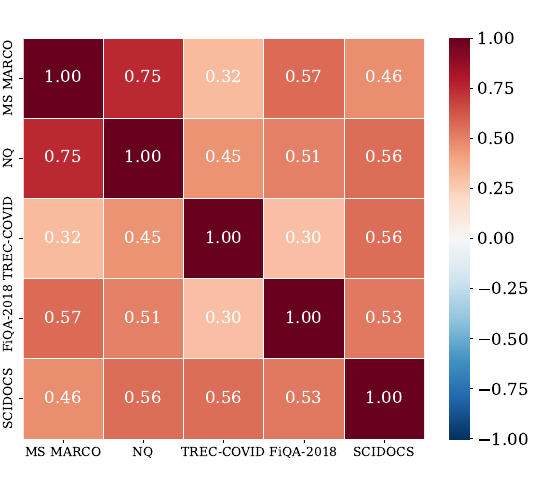}
      \caption{Cross-dataset consistency.}
      \label{fig:rq2-b}
    \end{subfigure}
  \end{tabular}
    \caption{
    The LLM–Human distinction forms a stable embedding-space direction. The plots demonstrate this consistency along three dimensions: (a) within datasets, (b) across datasets, and (c) across retrievers. Panels (a) and (b) use the DRAGON retriever. Results for all 14 datasets are in Appendix~\ref{app:embedding_full}.
    }
  \label{fig:rq2}
\end{figure*}

\begin{figure*}[t]
  \centering
  \begin{tabular}{
    @{}
    b{0.46\linewidth}
    @{\hspace{0.05\linewidth}}
    b{0.46\linewidth} 
    @{} 
  }

  \begin{subfigure}[t]{\linewidth}
    \centering
    \includegraphics[width=0.87\linewidth]{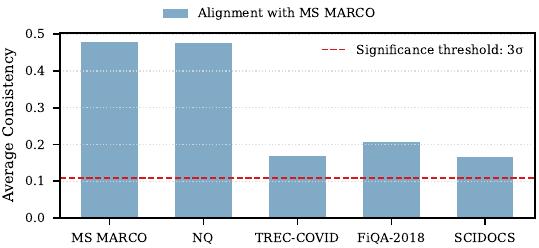}
    \caption{Cross-dataset consistency.}
    \label{fig:rq2_2-a}
  \end{subfigure}

  &
  \begin{subfigure}[t]{\linewidth}
    \centering
    \includegraphics[width=0.87\linewidth]{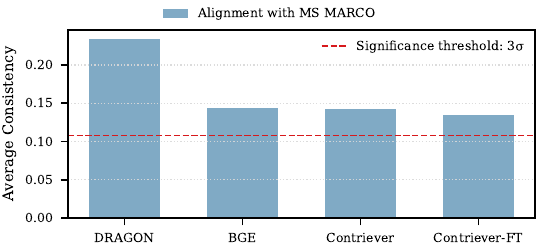}
    \caption{Cross-retriever consistency.}
    \label{fig:rq2_2-b}
  \end{subfigure}

  \end{tabular}

  \caption{
  The LLM--Human displacement aligns with the positive--negative supervision direction.
  Panels (a) and (b) show cross-dataset and cross-retriever consistency, respectively.
  Across both settings, cosine similarities exceed the 3$\sigma$ threshold,
  confirming a stable and coherent embedding-space direction.
  }
  \label{fig:rq2_2}
\end{figure*}

In this section, we examine whether the embedding shift associated with retrieval supervision
(positives vs. negatives) aligns with the shift associated with source type
(LLM-generated vs. human-written passages). Such alignment would suggest that
source-related differences overlap with representational directions learned from relevance
supervision. We proceed in three steps: (1) estimate a positive--negative supervision direction;
(2) estimate the LLM--human source direction and assess its consistency; and (3) evaluate whether the two directions are aligned.

\paragraph{Notation.}
Let $q$ denote a query and $d$ a passage. For supervised retrieval, $d^{+}$ and $d^{-}$ denote
an annotated positive passage and a sampled negative passage, respectively. For source-type
analysis, $d^{\mathrm{LLM}}$ and $d^{\mathrm{Human}}$ denote an LLM-generated passage and its
human-written counterpart. The query and document encoders $h_q(\cdot)$ and $h_d(\cdot)$ map
inputs to embeddings in $\mathbb{R}^m$, and the retrieval score is
\[
s_\theta(q,d)=\langle h_q(q), h_d(d)\rangle .
\]
We use $\delta$ to denote an embedding displacement vector. For example, the LLM--human
displacement is
\[
\delta^{\mathrm{LH}}
=
h_d(d^{\mathrm{LLM}})
-
h_d(d^{\mathrm{Human}}).
\]
The symbol $\bar{\delta}$ denotes the average displacement over a set of paired passages, and
$\mathbb{E}[\cdot]$ denotes expectation over the indicated distribution.

\paragraph{Estimating the Positive--Negative Embedding Direction.}
To estimate the embedding direction associated with retrieval supervision, we construct
positive--negative pairs from MS~MARCO. Rather than comparing annotated positives with arbitrary
unannotated passages, we follow common practice~\cite{karpukhinDense2020}: for each query, we
retrieve the top-10 BM25 candidates and randomly sample one as the negative paired with the
annotated positive. This produces a retrieval-like contrast in which positives and negatives have
substantial lexical and topical overlap. We then estimate the average positive--negative
displacement as
\[
\overline{\delta}_{\mathrm{PN}}
=
\mathbb{E}\big[h_d(d^+) - h_d(d^-)\big].
\]
We refer to this vector as the supervision-induced positive--negative direction.

\paragraph{Reference Baselines in High-Dimensional Space.}
To interpret the magnitude of directional consistency, we compare cosine similarities
with those of random vectors in the same dimensionality. In 768 dimensions, random
vectors are nearly orthogonal, so their cosine similarities concentrate around zero. Over
$99.7\%$ of random pairs fall within $\pm 3\sigma$ of the mean
(Appendix~\ref{app:cosine_null}). We therefore use this range as a descriptive reference: cosine similarities outside it indicate unusually strong alignment compared with random directions.

\paragraph{Does the LLM--Human Contrast Exhibit Directional Consistency?}
The LLM--human comparison uses semantically aligned counterparts, allowing us to compute
pairwise displacements directly:
\[
\delta_i^{\mathrm{LH}}
=
h_d(d_i^{\mathrm{LLM}})
-
h_d(d_i^{\mathrm{Human}}).
\]
We evaluate whether these displacements form a consistent direction along three dimensions.

\noindent\textbf{(1) Within datasets.}
We compute the average pairwise cosine similarity
\[
\mathbb{E}_{i\neq j}
\left[
\cos(\delta_i^{\mathrm{LH}}, \delta_j^{\mathrm{LH}})
\right].
\]
Values exceeding the random-vector reference range indicate stronger-than-random directional
consistency within a dataset, as shown in Figure~\ref{fig:rq2-a}.

\noindent\textbf{(2) Across datasets.}
For each dataset $D$, we compute the dataset-level mean displacement
\[
\overline{\delta}_{\mathrm{LH},D}
=
\mathbb{E}_{d_i\in D}
[
\delta_i^{\mathrm{LH}}
],
\]
and evaluate cross-dataset alignment by
\[
\cos(
\overline{\delta}_{\mathrm{LH},D_1},
\overline{\delta}_{\mathrm{LH},D_2}
).
\]
As shown in Figure~\ref{fig:rq2-b}, the mean source displacements are consistently aligned across
datasets.

\noindent\textbf{(3) Across retrievers.}
We repeat the analysis with multiple retrievers. Figure~\ref{fig:rq2-c} shows that the
LLM--human displacement exhibits similar directional patterns across the retrievers examined.

Together, these results show that the LLM--human contrast corresponds to a reproducible embedding
displacement across datasets and retrievers. The sign-flip control further shows that this
pattern reflects a shared direction rather than only non-directional properties such as vector
magnitude.

\paragraph{Do the Positive--Negative and LLM--Human Directions Align?}
We next examine whether the source-related displacement aligns with the positive--negative
supervision contrast. For each dataset $D$, we compute the cosine similarity between the mean
LLM--human displacement, $\overline{\delta}_{\mathrm{LH},D}$, and the positive--negative direction
estimated from MS~MARCO, $\overline{\delta}_{\mathrm{PN}}$. As shown in
Figure~\ref{fig:rq2_2-a}, the two directions are consistently aligned across datasets and exceed
the random-vector baseline. Figure~\ref{fig:rq2_2-b} shows a similar pattern across
retrievers, indicating that the alignment is not confined to a single model.

\paragraph{Empirical Validation Check.}
Because dense retriever embeddings can be anisotropic, the random-vector baseline serves mainly as an idealized geometric reference. As an empirical complement, we also report a paired sign-flip check in Appendix~\ref{app:embedding_full}. This check randomly flips the sign of each LLM--human displacement and recomputes the within-dataset consistency, cross-dataset consistency, and PN--LH alignment metrics, serving as a sign-randomized empirical baseline under the observed embedding geometry. The results remain consistent with our main findings.

This alignment indicates that the LLM--human source contrast and the positive--negative
supervision contrast share a representational component. Rather than treating this component as purely stylistic, we interpret it as a source-related signal that can overlap with relevance supervision in the learned embedding space. We next provide a theoretical perspective on how such signals can enter the retrieval score under contrastive training.

\begin{table}[!t]
  \centering
  \caption{$\Delta$NDSR@5 under negative-sampling strategies with increasing expected artifact mismatch (in-batch only $\rightarrow$ standard $\rightarrow$ hard-neg only). \colorbox{llmred}{Darker} Average-row shading denotes stronger LLM preference.}
  \label{tab:rq3_train}

  \begin{minipage}[t]{0.88\linewidth}
    \centering
    \renewcommand{\arraystretch}{1.0}
    \setlength{\tabcolsep}{4pt}
    \resizebox{\linewidth}{!}{%
      \begin{tabular}{l ccc}
        \toprule
        & In-batch only & Standard & Hard-neg only \\
        \midrule
        MS~MARCO      & 0.014  & -0.051 & -0.057 \\
        DL19          & 0.025  & -0.155 & -0.182 \\
        DL20          & 0.041  & -0.120 & -0.152 \\
        NQ            & 0.020  & -0.081 & -0.085 \\
        NFCorpus      & -0.050 & -0.068 & -0.093 \\
        TREC-COVID    & -0.182 & -0.252 & -0.285 \\
        HotpotQA      & 0.003  & 0.017  & -0.021 \\
        FiQA-2018     & -0.055 & -0.227 & -0.238 \\
        Touché-2020   & -0.077 & -0.202 & -0.193 \\
        DBpedia       & -0.021 & -0.041 & -0.043 \\
        SCIDOCS       & 0.010  & -0.051 & -0.035 \\
        FEVER         & 0.014  & -0.005 & -0.013 \\
        Climate-FEVER & -0.032 & -0.071 & -0.080 \\
        SciFact       & -0.032 & -0.051 & -0.053 \\
        \midrule
        \textbf{Average} &
          \avgshadetrain{-0.023}{-0.109} &
          \avgshadetrain{-0.097}{-0.109} &
          \avgshadetrain{-0.109}{-0.109} \\
        \bottomrule
      \end{tabular}
    }
  \end{minipage}
  \hfill
  \begin{minipage}[t]{0.08\linewidth}
    \centering
    \makebox[\linewidth][r]{\raisebox{-0.5\height}{\traincolorbarvertical[1.5]}}
  \end{minipage}
\end{table}

\begin{table*}[!t]
\centering
\setlength{\abovecaptionskip}{2pt}
\setlength{\belowcaptionskip}{2pt}
\caption{
$\Delta$NDSR@5 before and after inference-time debiasing across 5 datasets and 5 relevance-supervised retrievers, using leave-one-dataset-out calibration to estimate the debiasing direction.
Positive values indicate human preference, while negative values indicate LLM preference.
In the Average row, the first line reports mean $\Delta\mathrm{NDSR}@5$, and the second line reports the remaining $|\Delta\mathrm{NDSR}@5|$ relative to the original value.
Darker shading indicates stronger source bias.
Full results on all 14 datasets appear in Appendix~\ref{app:rq3_results}.
}
\label{tab:rq3_debias}

\resizebox{\textwidth}{!}{
\begin{minipage}{\textwidth}
\renewcommand{\arraystretch}{0.85}
\setlength{\tabcolsep}{3.5pt}
\resizebox{\linewidth}{!}{
\begin{tabular}{l | cc | cc | cc | cc | cc}
\toprule
\multirow{2}{*}{Dataset (↓)}
  & \multicolumn{2}{c}{ANCE}
  & \multicolumn{2}{c}{coCondenser}
  & \multicolumn{2}{c}{DRAGON}
  & \multicolumn{2}{c}{RetroMAE}
  & \multicolumn{2}{c}{TAS-B} \\
\cmidrule(lr){2-3} \cmidrule(lr){4-5} \cmidrule(lr){6-7} \cmidrule(lr){8-9} \cmidrule(lr){10-11}
 & Original & Debias & Original & Debias & Original & Debias & Original & Debias & Original & Debias \\
\midrule
MS~MARCO       & -0.042 & -0.004 & -0.020 & -0.002 & -0.083 & -0.035 & -0.083 & -0.053 & -0.121 & -0.064 \\
TREC-COVID     & -0.162 & -0.035 & -0.340 & -0.180 & -0.134 & -0.106 & -0.194 & -0.082 & -0.328 & -0.224 \\
NQ             & -0.042 &  0.003 & -0.072 & -0.030 & -0.099 & -0.050 & -0.060 & -0.029 & -0.078 & -0.055 \\
FiQA-2018      & -0.179 &  0.006 & -0.219 & -0.142 & -0.161 & -0.123 & -0.205 & -0.127 & -0.170 & -0.127 \\
SCIDOCS        & -0.040 & -0.017 & -0.058 & -0.017 & -0.048 & -0.035 & -0.073 & -0.002 & -0.054 & -0.021 \\
\midrule
\raisebox{1.5ex}{\textbf{Average}} &
  \shortstack[c]{\avgshadeinfer{-0.093}{-0.093}\\ (100\%)} &
  \shortstack[c]{\avgshadeinfer{-0.009}{-0.093}\\ (10\%)}  &
  \shortstack[c]{\avgshadeinfer{-0.142}{-0.142}\\ (100\%)} &
  \shortstack[c]{\avgshadeinfer{-0.074}{-0.142}\\ (52\%)}  &
  \shortstack[c]{\avgshadeinfer{-0.105}{-0.105}\\ (100\%)} &
  \shortstack[c]{\avgshadeinfer{-0.070}{-0.105}\\ (67\%)}  &
  \shortstack[c]{\avgshadeinfer{-0.123}{-0.123}\\ (100\%)} &
  \shortstack[c]{\avgshadeinfer{-0.059}{-0.123}\\ (48\%)}  &
  \shortstack[c]{\avgshadeinfer{-0.150}{-0.150}\\ (100\%)} &
  \shortstack[c]{\avgshadeinfer{-0.098}{-0.150}\\ (65\%)} \\
\bottomrule
\end{tabular}}
\end{minipage}
\hfill
\begin{minipage}{0.06\textwidth}
\centering
\debiascolorbarvertical[2] %
\end{minipage}
}
\end{table*}

\subsection{Theoretical Perspective}
\label{subsec:rq2_theory}

We present a compact theoretical account of how contrastive retrievers learn source-related components that align with relevance supervision. Formal statements and proofs are provided in Appendix~\ref{app:theory_and_proofs}.

Let a document $d$ be described by semantic features $M_d$ and artifact features
$A_d$, such as fluency and lexical specificity. These artifacts need not be
reliable indicators of query-specific relevance, but they can still influence
training if they help predict the positive--negative contrast beyond semantic
relevance alone. In the density-ratio view of contrastive learning, the learned score is related to the log ratio between positive-pair and negative-pair densities (Appendix~\ref{app:theory_and_proofs}). Thus, any feature predictive of this contrast, including artifact features, can contribute to the score, leading to:
\[
s^*(q,d)
\approx
\mathrm{Score}_{\mathrm{sem}}(q,M_d)
+
\mathrm{Score}_{\mathrm{art}}(q,A_d).
\]
The additive form can be viewed as a first-order Taylor approximation of the score around the observed data: the artifact term captures the leading local contribution of artifact features, while higher-order interactions are left in
the residual. Thus, when artifact features help distinguish positives from negatives in the training distribution, they can become \textbf{learnable shortcuts.}

For dot-product retrievers, this shortcut can appear as an embedding-space component. Locally linearizing the document encoder with respect to artifact features gives
\begin{align}
s_\theta(q,d)
& =
\langle h_q(q),h_d(d)\rangle \nonumber \\
& \approx \langle h_q(q),h_d^{\mathrm{sem}}(d)\rangle + \langle h_q(q),h_d^{\mathrm{art}}(d)\rangle . \nonumber
\end{align}

This view is consistent with our empirical finding that the LLM--human displacement is directionally stable and aligned with the positive--negative supervision direction. It also provides a basis for mitigation strategies that
reduce artifact imbalance in training or suppress artifact-aligned embedding components at inference time.
\section{RQ3: How Can Source Bias Be Mitigated?}
\label{sec:rq3}

The analyses above show that relevance supervision can turn source-aligned artifacts into shortcuts for relevance prediction. This mechanism yields two direct interventions. At training time, we can align the artifact distributions of positives and negatives so that source cues no longer separate the two pools. At inference time, we can suppress embedding components aligned with the LLM--human displacement identified in Section~\ref{subsec:rq2_embedding}. Both test the mechanism by weakening this signal and measuring its effect on source preference.

\paragraph{Training-time Intervention: Reducing Artifact Mismatch.}
Using the BEIR pipeline~\cite{devlinBERT2019,thakurBEIR2021a}, we fine-tune BERT retrievers on MS~MARCO while varying only negative sampling. \emph{In-batch only} draws negatives from passages annotated as positive for other queries, so both sides of the training contrast come from the annotated-positive pool and have the lowest expected artifact mismatch. \emph{Standard} adds one corpus-mined hard negative per query, whereas \emph{hard-neg only} uses the mined passage alone. Together, these settings form a spectrum of increasing expected artifact mismatch.

The proposed mechanism predicts increasingly stronger LLM preference from in-batch only to standard to hard-neg only. Table~\ref{tab:rq3_train} confirms this prediction. In-batch-only sampling yields the weakest average LLM preference across datasets ($\Delta\mathrm{NDSR}@5=-0.023$), while standard and hard-neg-only sampling produce substantially stronger preferences ($-0.097$ and $-0.109$). The same ordering holds across three random seeds using mean absolute $\Delta$NDSR@5 (Appendix Table~\ref{tab:rq3_multiseed}). The weaker source preference under in-batch-only sampling comes with lower NDCG@5 than the settings using mined hard negatives (Appendix Table~\ref{tab:rq3_train_performance_app}). Appendix~\ref{app:neg-sampling} examines whether this pattern is driven by hard-negative mining or by the composition of the negative pool.

\paragraph{Inference-time Intervention: Suppressing the Source-aligned Direction.}
We next examine whether source bias can be reduced without retraining. 
As shown in Section~\ref{subsec:rq2_embedding}, LLM-generated passages exhibit a consistent shift in embedding space relative to human-written ones.
Let
\[
n = \frac{\overline{\delta}_{\mathrm{LH}}}{\|\overline{\delta}_{\mathrm{LH}}\|}
\]
denote the normalized mean LLM--human displacement. 
Given a passage embedding $v=h_d(d)$, we suppress its component along this direction:
\[
v' = v - \langle v,n\rangle n .
\]

We estimate $n$ through \emph{leave-one-dataset-out calibration}, treating MS~MARCO, DL19, and DL20 as one group because they share the same collection. For each target dataset, we sample 100 human--LLM passage pairs from each of the remaining dataset groups, yielding 1100 pairs in total, and average their displacement vectors. This protocol leverages the cross-dataset consistency of the LLM--human displacement and yields a transferable estimate of the source-aligned direction.

We apply this projection to five relevance-supervised retrievers. As shown in Table~\ref{tab:rq3_debias}, it reduces source preference across models. Once calibrated, our method requires one linear projection per passage embedding.

The intervention also transfers across generator families. Replacing the Llama-2 passages with Qwen3-32B-generated counterparts, the projection reduces $|\Delta\mathrm{NDSR}@5|$ in 23 of 25 cells across the same five datasets and retrievers (Appendix Table~\ref{tab:rq3_qwen_app}). These results demonstrate effective debiasing and cross-generator generalization.

\FloatBarrier

\paragraph{Summary.}
At training time, LLM preference strengthens as negative sampling introduces greater artifact mismatch. At inference time, suppressing the LLM--human-aligned embedding component reduces source bias and transfers across generator families. These complementary interventions validate the proposed mechanism.

\section{Conclusion}

We provide a data-centric account of source bias in neural retrieval, showing that relevance supervision can inadvertently encourage retrievers to prefer LLM-generated passages. Our analyses suggest that this behavior arises in part from surface-level artifact imbalances in supervised retrieval data: positives and negatives differ in properties such as fluency and lexical specificity, and these differences mirror those between LLM-generated and human-written passages. Embedding-space evidence further shows that the positive--negative supervision direction aligns with the stable LLM--human shift, suggesting that artifact differences learned from the training contrast can favor LLM-generated passages at inference time.

We test this mechanism through two interventions. Reducing artifact mismatch during negative sampling and removing a source-aligned embedding component at inference both reduce source preference. Together, these findings show that source bias is not merely a model-side phenomenon, but is also shaped by relevance supervision, highlighting artifact-aware data construction and lightweight embedding-level correction as promising directions for robust retrieval in mixed human--LLM corpora.

\section*{Limitations}
Our artifact analysis is necessarily incomplete. PPL and IDF are only two interpretable examples, and source-related artifacts may be diverse, entangled, and difficult to enumerate exhaustively. Moreover, perfectly controlling for these artifacts in retrieval training is extremely challenging, as changing negatives may also change hardness, topical overlap, and false-negative rates. Therefore, although our model comparisons, fine-tuning experiments, and negative-sampling interventions support the role of relevance supervision, they cannot fully exclude other co-varying factors. Our results suggest that relevance supervision is a plausible contributor to source bias, rather than its sole cause.
\section*{Ethical Considerations}
All experiments use public datasets and models and do not involve personal information. Although source bias may affect content exposure, our work focuses on understanding its mechanism, and the direct ethical risks of our experiments are limited.
\section*{Acknowledgments}
This work was funded by the National Natural Science Foundation of China (NSFC) under Grants No. 62302486 and No. 62441229, the Innovation Project of ICT CAS under Grants No. E361140 and E561010, the Strategic Priority Research Program of the Chinese Academy of Sciences (CAS) under Grants No. XDB0680101 and No. XDB0680102, the CAS Project for Young Scientists in Basic Research under Grant No. YSBR-034, the National Key Research and Development Program of China under Grant No. 2023YFA1011602, and the Xiaomi Young Talents Program.

\bibliography{emnlp2026}

\appendix
\section{The Use of AI Writing Assistants}

In this study, we used AI writing assistants solely to improve the clarity and readability of the text. The models/tools were not used for research ideation, literature review, methodological development, experimental design, data analysis, result generation, or interpretation. All scientific content, analyses, conclusions, and responsibility for the work are entirely those of the authors.
\section{Reproducibility Resources}
\label{app:resources}

To facilitate reproducibility, we provide the datasets, model checkpoints, and main implementation details used in this work. 
All datasets and models are obtained from publicly available HuggingFace releases or their official websites. 
Our usage follows the respective licenses and research-only terms of the original sources. 
Unless otherwise stated, we use the official dataset splits and evaluation files provided by Cocktail and follow the standard BEIR-style dense retrieval evaluation protocol.

\subsection{Datasets}
\label{app:datasets}

Table~\ref{tab:datasets_links} lists the Cocktail datasets used in our experiments. 
Each Cocktail dataset augments a BEIR collection with LLM-generated counterparts for human-written passages. 
For source-preference evaluation, we retrieve from the mixed human--LLM corpus and compute $\Delta$NDSR@5 over the ranked results. 
For retrieval-effectiveness evaluation, we report NDCG@5 using the relevance judgments provided with the benchmark.

\begin{table*}[ht]
    \centering
    \caption{Datasets used in this paper (Cocktail versions) and their HuggingFace links.}
    \label{tab:datasets_links}
    \small
    \renewcommand{\arraystretch}{1.05}
    \setlength{\tabcolsep}{4pt}
    \begin{tabular}{@{}p{0.31\textwidth}@{\hspace{0.02\textwidth}}p{0.67\textwidth}@{}}
        \toprule
        \textbf{Dataset} & \textbf{HuggingFace Link} \\
        \midrule
        MS~MARCO~\citep{nguyen2016msmarco} & \url{https://huggingface.co/datasets/IR-Cocktail/msmarco} \\
        TREC-DL’19~\citep{craswellOverview2020} & \url{https://huggingface.co/datasets/IR-Cocktail/dl19} \\
        TREC-DL’20~\citep{craswellOverview2021} & \url{https://huggingface.co/datasets/IR-Cocktail/dl20} \\
        Natural Questions~\citep{kwiatkowskiNatural2019} & \url{https://huggingface.co/datasets/IR-Cocktail/nq} \\
        NFCorpus~\citep{botevaFullText2016} & \url{https://huggingface.co/datasets/IR-Cocktail/nfcorpus} \\
        TREC-COVID~\citep{voorheesTRECCOVID2021} & \url{https://huggingface.co/datasets/IR-Cocktail/trec-covid} \\
        HotpotQA~\citep{yangHotpotQA2018a} & \url{https://huggingface.co/datasets/IR-Cocktail/hotpotqa} \\
        FiQA-2018~\citep{maiaWWW182018a} & \url{https://huggingface.co/datasets/IR-Cocktail/fiqa} \\
        Touché-2020~\citep{bondarenko2020overview} & \url{https://huggingface.co/datasets/IR-Cocktail/webis-touche2020} \\
        DBpedia~\citep{hasibiDBpediaEntity2017} & \url{https://huggingface.co/datasets/IR-Cocktail/dbpedia-entity} \\
        SCIDOCS~\citep{cohanSPECTER2020a} & \url{https://huggingface.co/datasets/IR-Cocktail/scidocs} \\
        FEVER~\citep{thorneFEVER2018a} & \url{https://huggingface.co/datasets/IR-Cocktail/fever} \\
        Climate-FEVER~\citep{diggelmannCLIMATEFEVER2021a} & \url{https://huggingface.co/datasets/IR-Cocktail/climate-fever} \\
        SciFact~\citep{waddenFact2020} & \url{https://huggingface.co/datasets/IR-Cocktail/scifact} \\
        \bottomrule
    \end{tabular}
\end{table*}

\subsection{Model Checkpoints}
\label{app:models}

Table~\ref{tab:models_links} lists the dense retriever checkpoints used in our experiments. 
For models with separate query and passage encoders, such as DRAGON, we use the corresponding released encoders. For single-encoder models, the same encoder is used for both queries and passages. We follow the official usage of each checkpoint, including its tokenizer, pooling strategy, embedding normalization, and input prefixes when applicable. For example, E5-style models are evaluated with their official query and passage prefixes.

\begin{table*}[ht]
    \centering
    \caption{Dense retriever checkpoints used in this paper and their HuggingFace links.}
    \label{tab:models_links}
    \small
    \renewcommand{\arraystretch}{1.05}
    \setlength{\tabcolsep}{4pt}
    \begin{tabular}{@{}p{0.31\textwidth}@{\hspace{0.02\textwidth}}p{0.67\textwidth}@{}}
        \toprule
        \textbf{Model} & \textbf{HuggingFace Link} \\
        \midrule
        \multicolumn{2}{l}{\textit{Relevance-Supervised Retrievers}} \\
        ANCE~\citep{xiongApproximate2020} & \url{https://huggingface.co/sentence-transformers/msmarco-roberta-base-ance-firstp} \\
        TAS-B~\citep{hofstatterEfficiently2021} & \url{https://huggingface.co/sentence-transformers/msmarco-distilbert-base-tas-b} \\
        coCondenser~\citep{gaoUnsupervised2022} & \url{https://huggingface.co/sentence-transformers/msmarco-bert-co-condensor} \\
        RetroMAE~\citep{xiaoRetroMAE2022} & \url{https://huggingface.co/nthakur/RetroMAE_BEIR} \\
        DRAGON (query encoder)~\citep{linHow2023} & \url{https://huggingface.co/nthakur/dragon-plus-query-encoder} \\
        DRAGON (corpus encoder)~\citep{linHow2023} & \url{https://huggingface.co/nthakur/dragon-plus-context-encoder} \\
        \midrule
        \multicolumn{2}{l}{\textit{General-Purpose Embedding Models}} \\
        BGE-base~\citep{xiaoCPack2024} & \url{https://huggingface.co/BAAI/bge-base-en-v1.5} \\
        BCE~\citep{youdao_bcembedding_2023} & \url{https://huggingface.co/maidalun1020/bce-embedding-base_v1} \\
        GTE~\citep{liGeneral2023} & \url{https://huggingface.co/thenlper/gte-base} \\
        E5~\citep{wangText2024} & \url{https://huggingface.co/intfloat/e5-base-v2} \\
        M3E~\citep{Moka_Massive_Mixed_Embedding} & \url{https://huggingface.co/moka-ai/m3e-base} \\
        \midrule
        \multicolumn{2}{l}{\textit{Unsupervised Retrievers}} \\
        Contriever~\citep{izacardUnsupervised2022} & \url{https://huggingface.co/nishimoto/contriever-sentencetransformer} \\
        E5-unsup~\citep{wangText2024} & \url{https://huggingface.co/intfloat/e5-base-unsupervised} \\
        SimCSE~\citep{gaoSimCSE2021} & \url{https://huggingface.co/princeton-nlp/unsup-simcse-bert-base-uncased} \\
        \bottomrule
    \end{tabular}
\end{table*}

\subsection{Implementation Details}
\label{app:implementation_details}

For all dense retrieval experiments, we encode queries and passages independently and follow the usage recommended by each released checkpoint for similarity computation, pooling, embedding normalization, and input prefixes when applicable. The maximum query and passage lengths are set to 32 and 512 tokens, respectively.

For the linguistic diagnostics in RQ2, the main analysis computes passage PPL with Llama-3-8B-Instruct. The cross-family validation replaces the scorer with Qwen3-8B while keeping the passage set and PPL definition fixed (Appendix Table~\ref{tab:qwen_ppl}). For the artifact--score analysis, we z-normalize document scores separately within each query and compute Pearson correlations with document-level PPL and IDF over approximately 57.6 million query--document pairs (Table~\ref{tab:artifact_score_corr}).

For the fine-tuning experiments in RQ1, we fine-tune Contriever, E5-unsup, and SimCSE on MS~MARCO using the same contrastive retrieval objective and training pipeline. The loss is InfoNCE with dot-product similarity. Following the BEIR training script, we train for 10 epochs using AdamW with a learning rate of $2\times10^{-5}$, a batch size of 75, and 1,000 warmup steps. The remaining optimizer settings follow the SentenceTransformers defaults. 
The main results use random seed 42, and the stability analysis repeats the experiment with seeds 100 and 2026 (Appendix Table~\ref{tab:rq1_multiseed}). The same maximum query and passage lengths as in evaluation are used during fine-tuning.

For the negative-sampling intervention in RQ3, we keep the model architecture, loss form, optimizer, learning rate, batch size, and training steps fixed, and vary only the negative sampling strategy. The \emph{in-batch only} setting uses other queries' annotated positives in the same mini-batch as negatives. The \emph{standard} setting combines in-batch negatives with one mined hard negative per query. The \emph{hard-neg only} setting uses one mined hard negative per query without in-batch negatives; this is implemented with a custom loss that computes the contrastive objective only over each query's positive and sampled hard negative. Standard and hard-neg-only negatives are sampled from the official MS~MARCO hard-negative file. The three main-text conditions use seed 42 and are repeated with seeds 100 and 2026 for the stability analysis (Appendix Table~\ref{tab:rq3_multiseed}). For the additional control reported in Appendix~\ref{app:neg-sampling:aligned-negatives}, BM25 instead searches passages annotated as relevant to other training queries and selects the highest-ranked passage that is not the current query's ground truth.

For the length-controlled evaluation, we use Qwen2.5-7B-Instruct to generate meaning-preserving longer variants with the prompt in Figure~\ref{fig:rephrasing_prompt}.

For inference-time debiasing, we follow the leave-one-dataset-out calibration protocol in Section~\ref{sec:rq3}, sampling 100 human--LLM passage pairs from each non-target dataset group to estimate the LLM--human direction. The projection is applied only to passage embeddings, without post-projection renormalization. For the baseline comparison, we implement CDC~\citep{wangPerplexity2024a} with document-level PPL as its source-related scalar signal and evaluate it from the same original retriever scores (Appendix Table~\ref{tab:rq3_cdc_app}). For cross-generator evaluation, we replace the Llama-2-generated passages with Qwen3-32B~\citep{qwen3technicalreport} counterparts on the five datasets used in the main RQ3 table and rerun the same leave-one-dataset-out projection protocol (Appendix Table~\ref{tab:rq3_qwen_app}).

Unless otherwise noted, all experiments are conducted on a single NVIDIA RTX A6000 GPU with 48 GB of memory.
\section{Dataset Statistics}
\label{app:dataset_stats}

Table~\ref{tab:dataset_stat} summarizes the statistics of the 14 datasets used in this paper. 
This table is adapted from the Cocktail benchmark~\citep{daiCocktail2024a}, with minor modifications. 

\begin{table*}[t]
\centering
\caption{Statistics of the 14 datasets in the Cocktail benchmark used in this paper. 
Avg. D/Q denotes the average number of relevant documents per query. 
This table is adapted from \citet{daiCocktail2024a}.}
\label{tab:dataset_stat}
\renewcommand{\arraystretch}{1.2}
\resizebox{\textwidth}{!}{
\begin{tabular}{lcccccccc}
\toprule
\textbf{Dataset} & \textbf{Domain} & \textbf{Task} & \textbf{Relevancy} & \textbf{\#Pairs} & \textbf{\#Queries} & \textbf{\#Corpus} & \textbf{Avg. D/Q} & \textbf{Avg. Length (Q / Human / LLM)} \\
\midrule
MS~MARCO      & Misc.       & Passage Retrieval     & Binary   & 532,663  & 6,979  & 542,203  & 1.1   & 6.0 / 58.1 / 55.1 \\
DL19          & Misc.       & Passage Retrieval     & Binary   & -        & 43     & 542,203  & 95.4  & 5.4 / 58.1 / 55.1 \\
DL20          & Misc.       & Passage Retrieval     & Binary   & -        & 54     & 542,203  & 66.8  & 6.0 / 58.1 / 55.1 \\
TREC-COVID    & Biomedical  & Biomedical IR         & 3-level  & -        & 50     & 128,585  & 430.1 & 10.6 / 197.6 / 165.9 \\
NFCorpus      & Biomedical  & Biomedical IR         & 3-level  & 110,575  & 323    & 3,633    & 38.2  & 3.3 / 221.0 / 206.7 \\
NQ            & Wikipedia   & QA                    & Binary   & -        & 3,446  & 104,194  & 1.2   & 9.2 / 86.9 / 81.0 \\
HotpotQA      & Wikipedia   & QA                    & Binary   & 169,963  & 7,405  & 111,107  & 2.0   & 17.7 / 67.9 / 66.6 \\
FiQA-2018     & Finance     & QA                    & Binary   & 14,045   & 648    & 57,450   & 2.6   & 10.8 / 133.2 / 107.8 \\
Touché-2020   & Misc.       & Argument Retrieval    & 3-level  & -        & 49     & 101,922  & 18.4  & 6.6 / 165.4 / 134.4 \\
DBpedia       & Wikipedia   & Entity Retrieval      & 3-level  & -        & 400    & 145,037  & 37.3  & 5.4 / 53.1 / 54.0 \\
SCIDOCS       & Scientific  & Citation Prediction   & Binary   & -        & 1,000  & 25,259   & 4.7   & 9.4 / 169.7 / 161.8 \\
FEVER         & Wikipedia   & Fact Checking         & Binary   & 140,079  & 6,666  & 114,529  & 1.2   & 8.1 / 113.4 / 91.1 \\
Climate-FEVER & Wikipedia   & Fact Checking         & Binary   & -        & 1,535  & 101,339  & 3.0   & 20.2 / 99.4 / 81.3 \\
SciFact       & Scientific  & Fact Checking         & Binary   & 919      & 300    & 5,183    & 1.1   & 12.4 / 201.8 / 192.7 \\
\bottomrule
\end{tabular}
}
\end{table*}
\section{Retrieval Effectiveness and Fine-Tuning Stability}
\label{app:rq1_effectiveness}

This section reports retrieval effectiveness on the Cocktail benchmark and the multi-seed stability of the relevance fine-tuning experiment.

\subsection{Retrieval Effectiveness}
For completeness, we report the retrieval effectiveness of all evaluated models on the Cocktail benchmark.
Table~\ref{tab:ndcg_main} presents NDCG@5 across 14 datasets for the 13 retrievers spanning the three model families. 
Table~\ref{tab:ndcg_ctrl} further reports results after fine-tuning unsupervised retrievers on MS~MARCO. 
These results complement the source preference analyses in Section~\ref{sec:rq1}. 

\begin{table*}[t]
\centering
\caption{NDCG@5 results across 14 datasets for 13 dense retrievers. Higher is better.}
\label{tab:ndcg_main}
\renewcommand{\arraystretch}{1.0}
\setlength{\tabcolsep}{3.5pt}
\resizebox{1.0\textwidth}{!}{
\begin{tabular}{lccccc | ccccc | ccc}
\toprule
\multirow{2}{*}{Dataset (↓)} & \multicolumn{5}{c}{Relevance-Supervised Retrievers} & \multicolumn{5}{c}{General-Purpose Embedding Models} & \multicolumn{3}{c}{Unsupervised Retrievers} \\
\cmidrule(lr){2-6} \cmidrule(lr){7-11} \cmidrule(lr){12-14}
 & ANCE & TAS-B & coCondenser & RetroMAE & DRAGON & BGE & BCE & GTE & E5 & M3E & Contriever & E5-unsup & SimCSE \\
\midrule
MS~MARCO       & 0.647 & 0.680 & 0.683 & 0.688 & 0.735 & 0.688 & 0.590 & 0.688 & 0.702 & 0.473 & 0.504 & 0.575 & 0.245 \\
DL19           & 0.686 & 0.760 & 0.734 & 0.743 & 0.771 & 0.755 & 0.708 & 0.750 & 0.747 & 0.507 & 0.515 & 0.624 & 0.346 \\
DL20           & 0.701 & 0.724 & 0.708 & 0.751 & 0.758 & 0.729 & 0.651 & 0.718 & 0.743 & 0.489 & 0.492 & 0.597 & 0.289 \\
NQ             & 0.640 & 0.708 & 0.711 & 0.746 & 0.790 & 0.778 & 0.625 & 0.789 & 0.790 & 0.494 & 0.623 & 0.737 & 0.353 \\
NFCorpus       & 0.266 & 0.340 & 0.345 & 0.336 & 0.389 & 0.403 & 0.275 & 0.394 & 0.368 & 0.257 & 0.339 & 0.371 & 0.109 \\
TREC-COVID     & 0.671 & 0.670 & 0.677 & 0.735 & 0.678 & 0.783 & 0.574 & 0.763 & 0.714 & 0.390 & 0.391 & 0.605 & 0.296 \\
HotpotQA       & 0.553 & 0.705 & 0.663 & 0.747 & 0.799 & 0.792 & 0.533 & 0.761 & 0.801 & 0.575 & 0.650 & 0.668 & 0.369 \\
FiQA-2018      & 0.275 & 0.408 & 0.467 & 0.498 & 0.529 & 0.384 & 0.285 & 0.380 & 0.373 & 0.366 & 0.225 & 0.373 & 0.093 \\
Touché-2020    & 0.479 & 0.427 & 0.349 & 0.441 & 0.390 & 0.402 & 0.333 & 0.423 & 0.411 & 0.242 & 0.308 & 0.333 & 0.252 \\
DBpedia        & 0.408 & 0.493 & 0.493 & 0.528 & 0.533 & 0.514 & 0.360 & 0.514 & 0.541 & 0.370 & 0.427 & 0.488 & 0.259 \\
SCIDOCS        & 0.095 & 0.111 & 0.102 & 0.116 & 0.123 & 0.177 & 0.118 & 0.190 & 0.141 & 0.069 & 0.114 & 0.174 & 0.041 \\
FEVER          & 0.820 & 0.835 & 0.842 & 0.870 & 0.876 & 0.928 & 0.682 & 0.924 & 0.905 & 0.865 & 0.878 & 0.925 & 0.510 \\
Climate-FEVER  & 0.270 & 0.306 & 0.255 & 0.311 & 0.318 & 0.368 & 0.274 & 0.373 & 0.303 & 0.161 & 0.223 & 0.264 & 0.195 \\
SciFact        & 0.465 & 0.602 & 0.549 & 0.611 & 0.631 & 0.715 & 0.533 & 0.732 & 0.688 & 0.448 & 0.614 & 0.719 & 0.239 \\
\bottomrule
\end{tabular}
}
\end{table*}

\begin{table}[t]
\centering
\caption{NDCG@5 results of unsupervised retrievers after MS~MARCO fine-tuning, corresponding to the same base models in Table~\ref{tab:ndcg_main}. The ``-FT'' suffix denotes fine-tuning on MS~MARCO.}
\label{tab:ndcg_ctrl}
\renewcommand{\arraystretch}{1.0}
\setlength{\tabcolsep}{4pt}
\resizebox{\columnwidth}{!}{%
\begin{tabular}{lccc}
\toprule
\multirow{2}{*}{Dataset} & \multicolumn{3}{c}{MS~MARCO Fine-Tuned Retrievers} \\
\cmidrule(lr){2-4}
& Contriever-FT & E5-FT & SimCSE-FT \\
\midrule
MS~MARCO      & 0.676 & 0.711 & 0.630 \\
DL19          & 0.696 & 0.763 & 0.727 \\
DL20          & 0.673 & 0.720 & 0.703 \\
NQ            & 0.732 & 0.764 & 0.670 \\
NFCorpus      & 0.339 & 0.378 & 0.279 \\
TREC-COVID    & 0.446 & 0.731 & 0.590 \\
HotpotQA      & 0.712 & 0.735 & 0.577 \\
FiQA-2018     & 0.255 & 0.336 & 0.220 \\
Touché-2020   & 0.347 & 0.428 & 0.389 \\
DBpedia       & 0.495 & 0.532 & 0.444 \\
SCIDOCS       & 0.117 & 0.138 & 0.083 \\
FEVER         & 0.857 & 0.895 & 0.837 \\
Climate-FEVER & 0.289 & 0.312 & 0.261 \\
SciFact       & 0.593 & 0.679 & 0.470 \\
\bottomrule
\end{tabular}%
}
\end{table}

\subsection{Multi-Seed Stability of Relevance Fine-Tuning}

We repeat the RQ1 fine-tuning experiment with random seeds 42, 100, and 2026. Table~\ref{tab:rq1_multiseed} reports the average absolute $\Delta$NDSR@5 across the 14 datasets for each run. All three retrievers exhibit low cross-seed variability, confirming that the supervision-induced source preference is stable across training runs.

\begin{table}[t]
\centering
\caption{Multi-seed stability of the RQ1 fine-tuning results. Each seed column reports mean absolute $\Delta$NDSR@5 across the 14 datasets.}
\label{tab:rq1_multiseed}
\renewcommand{\arraystretch}{1.0}
\setlength{\tabcolsep}{3.5pt}
\resizebox{\columnwidth}{!}{%
\begin{tabular}{lcccc}
\toprule
Model & Seed 42 & Seed 100 & Seed 2026 & Mean $\pm$ std \\
\midrule
Contriever-FT & 0.085 & 0.090 & 0.085 & $0.087 \pm 0.003$ \\
E5-FT         & 0.110 & 0.117 & 0.102 & $0.109 \pm 0.008$ \\
SimCSE-FT     & 0.091 & 0.080 & 0.082 & $0.084 \pm 0.006$ \\
\bottomrule
\end{tabular}%
}
\end{table}
\section{Formal Statements and Proofs}
\label{app:theory_and_proofs}

We formalize a mechanism by which artifact differences between positives and negatives can enter the contrastive retrieval objective. The analysis proceeds in three steps: (1) deriving the population-optimal scorer under InfoNCE, (2) decomposing the resulting density ratio into semantic and artifact-related contributions, and (3) relating this decomposition to a local embedding-space view for dot-product retrievers.

\paragraph{Notation and Setting.}
Let $q$ denote a query and $d$ a document. Each document $d$ is associated with semantic features $M_d$ and artifact features $A_d$, such as perplexity, lexical specificity, or stylistic attributes. We treat $(M_d,A_d)$ as random variables induced by a measurable feature map applied to documents.

A dense retriever consists of query and document encoders
$h_q(q),h_d(d)\in\mathbb{R}^m$, with a typical dot-product score
\begin{align*}
s_\theta(q,d)=\langle h_q(q),h_d(d)\rangle .
\end{align*}

Let $p_{\mathrm{pos}}(q,d)$ denote the distribution of positive query--document pairs. We first consider the idealized case in which negatives are sampled independently from a query-independent marginal distribution $p_{\mathrm{neg}}(d)=p(d)$. This setting matches the standard density-ratio analysis of InfoNCE and is closest to random or in-batch negative sampling. For query-dependent negative samplers, such as mined hard negatives, the same derivation yields a density ratio with $p_{\mathrm{neg}}(d\mid q)$ in the denominator.

We use the term \emph{artifact imbalance} to refer to empirical differences in artifact distributions between positives and negatives. A marginal diagnostic is
\[
\Delta_A=\mathbb{E}[A_{d^+}]-\mathbb{E}[A_{d^-}],
\]
but the theoretical condition below concerns artifact differences after conditioning on query and semantic features.

\paragraph{Step 1: Optimal scorer under InfoNCE.}
InfoNCE encourages the model to identify the positive passage among a set of candidates. Under population risk minimization with an unrestricted scoring function, its optimal score has a density-ratio form.

\begin{lemma}\label{lem:density}
For InfoNCE with negatives sampled independently from $p(d)$, the population-optimal unrestricted scorer is equivalent, up to an additive function of $q$, to
\[
s^{*}(q,d)
=
\log \frac{p_{\mathrm{pos}}(d\mid q)}{p(d)}
+
C(q).
\]
Equivalently, when the query marginal under positive-pair sampling is $p(q)$,
\[
s^{*}(q,d)
=
\log \frac{p_{\mathrm{pos}}(q,d)}{p(q)p(d)}
+
C(q).
\]
\end{lemma}

\begin{center}
\begin{tcolorbox}[width=1.0\linewidth, boxrule=0pt, top=3pt, bottom=3pt, colback=gray!20, colframe=gray!20]
\textbf{Insight 1:}
Under independent negative sampling and an unrestricted scorer, the population optimum of InfoNCE has a log-density-ratio form.
\end{tcolorbox}
\end{center}

\paragraph{Step 2: Decomposition into semantic and artifact contributions.}

We next decompose the density ratio in the feature space $(M_d,A_d)$. The decomposition separates a component depending only on the query and semantic features from a residual component that captures the additional contribution of artifacts conditional on those variables.

\begin{proposition}[Conditional artifact contribution]\label{prop:decom}
Let $(M_d,A_d)$ be semantic and artifact features induced by a measurable feature map. Consider the population-optimal scorer under a negative sampler $p_{\mathrm{neg}}(d\mid q)$:
\[
s^{*}(q,d)
=
\log \frac{p_{\mathrm{pos}}(d\mid q)}
{p_{\mathrm{neg}}(d\mid q)}
+
C(q).
\]
In the feature space $(M_d,A_d)$, the corresponding density ratio can be written as
\begin{align*}
s^{*}(q,d)
=
\phi(q,M_d)
+
\psi(q,M_d,A_d)
+
C(q),
\end{align*}
where $\phi(q,M_d)$ collects terms depending only on $(q,M_d)$, and $\psi(q,M_d,A_d)$ captures the conditional artifact contribution.

If $p_{\mathrm{pos}}(A_d\mid q,M_d)
\neq
p_{\mathrm{neg}}(A_d\mid q,M_d)
$
on a set of positive measure, then $\psi(q,M_d,A_d)$ is non-constant as a function of $A_d$ on that set. Consequently, the optimal scorer cannot in general be expressed as a function of $(q,M_d)$ alone.
\end{proposition}

\begin{center}
\begin{tcolorbox}[width=1.0\linewidth, boxrule=0pt, top=3pt, bottom=3pt, colback=gray!20, colframe=gray!20]
\textbf{Insight 2:}
When artifacts distinguish positives from negatives after conditioning on query and semantic content, artifact-dependent terms can enter the optimal contrastive scorer.
\end{tcolorbox}
\end{center}

\paragraph{Step 3: A local embedding-space view.}
We finally relate the preceding decomposition to dot-product retrievers. The goal is to show how artifact-related variation can appear as a local linear component in the document embedding.

\begin{proposition}[Local embedding decomposition]\label{prop:sem-art}
Let the document encoder $h_d$ be differentiable in a local feature chart $(M_d,A_d)$. Around a reference artifact value $a_0$, the dot-product score admits the first-order expansion
\begin{align*}
s_\theta(q,d)
&=
\underbrace{\langle h_q(q),h_{\mathrm{sem}}(M_d)\rangle}_{\text{semantic reference}} \\
&\quad
+
\underbrace{\langle h_q(q),h_{\mathrm{art}}(A_d;M_d)\rangle}_{\text{artifact perturbation}}
\\
&\quad
+
\varepsilon(q,M_d,A_d),
\end{align*}
where
\begin{align*}
&h_{\mathrm{sem}}(M_d)=g(M_d,a_0), \\
&h_{\mathrm{art}}(A_d;M_d)=J_a(M_d,a_0)(A_d-a_0),
\end{align*}
$g(M_d,A_d)$ denotes the document embedding as a function of the feature variables, $J_a(M_d,a_0)$ is the Jacobian of $g$ with respect to artifact features at $a_0$, and
\[
\varepsilon(q,M_d,A_d)=o(\|A_d-a_0\|)
\quad
\text{as}
\quad
\|A_d-a_0\|\to 0.
\]
If the artifact Jacobian varies weakly with $M_d$, the artifact perturbation can be approximated by a shared linear direction or low-dimensional subspace.
\end{proposition}

\begin{center}
\begin{tcolorbox}[width=1.0\linewidth, boxrule=0pt, top=3pt, bottom=3pt, colback=gray!20, colframe=gray!20]
\textbf{Insight 3:}
A local first-order approximation suggests that artifact-related variation can appear as approximately linear components in embedding space.
\end{tcolorbox}
\end{center}

Formal proofs of Lemma~\ref{lem:density}, Proposition~\ref{prop:decom}, and Proposition~\ref{prop:sem-art} are provided below. Together, these results characterize a mechanism through which relevance supervision can encode artifact-dependent signals. When such signals align with differences between LLM-generated and human-written passages, they can contribute to source preference. The embedding-space view further motivates our empirical search for stable source-aligned directions.

\subsection{Proof of Lemma~\ref{lem:density}}
\label{app:proofs}

\begin{proof}
We derive the population-optimal scorer for InfoNCE under independent negative sampling. The proof proceeds in three steps: (i) formalize the sampling and objective, (ii) show that risk minimization matches the true posterior over the positive index, and (iii) compute this posterior.

\textbf{Step 1: Sampling scheme and objective.}
Draw a query $q\sim p(q)$ and sample an index $I\sim\mathrm{Unif}\{0,\dots,K\}$, where $I$ denotes the position of the positive passage. Conditioned on $(q,I)$, sample the positive passage
$d_I\sim p_{\mathrm{pos}}(d\mid q)$ and sample negatives
$d_j\sim p(d)$ for all $j\neq I$. This yields a candidate batch
$\bm d=(d_0,\dots,d_K)$.

Given scores $s(q,d_j)$, the model predicts
\begin{equation*}
\pi_\theta(i\mid q,\bm d)
=
\frac{\exp(s(q,d_i))}
{\sum_{j=0}^{K}\exp(s(q,d_j))}.
\end{equation*}
A temperature parameter can be included by rescaling the score; it does not affect the density-ratio form of the optimum.

The InfoNCE loss is the expected negative log-likelihood:
\begin{align*}
\mathcal{L}(\theta)
=
\mathbb{E}_{q,\bm d}
\left[
\mathbb{E}_{I\mid q,\bm d}
\big[-\log \pi_\theta(I\mid q,\bm d)\big]
\right].
\end{align*}
Let
\[
P_i=\mathbb{P}(I=i\mid q,\bm d),
\qquad
\pi_i=\pi_\theta(i\mid q,\bm d).
\]
For fixed $(q,\bm d)$, the conditional risk is
\begin{equation*}
R(\bm s;q,\bm d)
=
-\sum_{i=0}^{K}P_i\log \pi_i.
\end{equation*}

\textbf{Step 2: Bayes optimality.}
The conditional risk can be written as
\begin{align*}
R(\bm s;q,\bm d)
&=
-\sum_i P_i\log P_i
+
\sum_i P_i\log\frac{P_i}{\pi_i}
\\
&=
H(P)+\mathrm{KL}(P\Vert \pi).
\end{align*}
Since $H(P)$ is independent of the model and $\mathrm{KL}(P\Vert \pi)\ge 0$, the risk is minimized iff $\pi=P$.

Because $\pi$ is a softmax over scores, any optimal scoring function satisfies
\begin{equation*}
s(q,d_i)=\log P_i+C(q,\bm d),
\end{equation*}
where $C(q,\bm d)$ is shared across all candidates in the batch and therefore does not affect the softmax.

\textbf{Step 3: Computing the posterior.}
By Bayes' rule and the sampling scheme,
\begin{align*}
P_i
&=
\mathbb{P}(I=i\mid q,\bm d)
\\
&\propto
\mathbb{P}(I=i)\,
p_{\mathrm{pos}}(d_i\mid q)
\prod_{j\neq i}p(d_j)
\\
&=
\frac{1}{K+1}
p_{\mathrm{pos}}(d_i\mid q)
\prod_{j\neq i}p(d_j).
\end{align*}
Normalizing over $i$ gives
\begin{equation*}
P_i
=
\frac{\displaystyle \frac{p_{\mathrm{pos}}(d_i\mid q)}{p(d_i)}}
{\displaystyle \sum_{j=0}^{K}
\frac{p_{\mathrm{pos}}(d_j\mid q)}{p(d_j)}}.
\end{equation*}
Therefore,
\begin{align*}
\log P_i
&=
\log\frac{p_{\mathrm{pos}}(d_i\mid q)}{p(d_i)}
-
\log
\sum_{j=0}^{K}
\frac{p_{\mathrm{pos}}(d_j\mid q)}{p(d_j)}.
\end{align*}
The second term is shared across all candidates in the batch and can be absorbed into the additive constant. Thus,
\begin{equation*}
s^*(q,d)
=
\log\frac{p_{\mathrm{pos}}(d\mid q)}{p(d)}
+
C(q).
\end{equation*}
If the positive-pair query marginal is $p(q)$, then
\[
p_{\mathrm{pos}}(d\mid q)
=
\frac{p_{\mathrm{pos}}(q,d)}{p(q)},
\]
and hence
\begin{equation*}
s^*(q,d)
=
\log\frac{p_{\mathrm{pos}}(q,d)}{p(q)p(d)}
+
C(q).
\end{equation*}
This completes the proof.

\paragraph{Remark.}
If negatives are drawn from a query-dependent distribution $p_{\mathrm{neg}}(d\mid q)$, the same derivation gives
\[
s^*(q,d)
=
\log
\frac{p_{\mathrm{pos}}(d\mid q)}
{p_{\mathrm{neg}}(d\mid q)}
+
C(q).
\]
Thus, the optimal score is a density ratio between the positive-pair distribution and the actual negative sampler.
\end{proof}

\subsection{Proof of Proposition~\ref{prop:decom}}

\begin{proof}
We prove the result for a general negative sampler $p_{\mathrm{neg}}(d\mid q)$. Let $(Q,M,A)$ denote the random variables induced by sampling $(q,d)$ and applying the feature map $d\mapsto(M_d,A_d)$. We write the positive and negative feature-space distributions as
\[
p_{\mathrm{pos}}(m,a\mid q),
\qquad
p_{\mathrm{neg}}(m,a\mid q).
\]
By Lemma~\ref{lem:density} and its query-dependent extension, the population-optimal score can be written, up to an additive function of $q$, as
\[
s^*(q,d)
=
\log
\frac{p_{\mathrm{pos}}(m_d,a_d\mid q)}
{p_{\mathrm{neg}}(m_d,a_d\mid q)}
+
C(q),
\]
where $(m_d,a_d)$ are the feature realizations of $d$.

Using the chain rule,
\begin{align*}
&\log
\frac{p_{\mathrm{pos}}(m_d,a_d\mid q)}
{p_{\mathrm{neg}}(m_d,a_d\mid q)}
\\
&=
\log
\frac{p_{\mathrm{pos}}(m_d\mid q)}
{p_{\mathrm{neg}}(m_d\mid q)}
+
\log
\frac{p_{\mathrm{pos}}(a_d\mid q,m_d)}
{p_{\mathrm{neg}}(a_d\mid q,m_d)}.
\end{align*}
Define
\begin{align*}
\phi(q,m_d)
&=
\log
\frac{p_{\mathrm{pos}}(m_d\mid q)}
{p_{\mathrm{neg}}(m_d\mid q)},
\\
\psi(q,m_d,a_d)
&=
\log
\frac{p_{\mathrm{pos}}(a_d\mid q,m_d)}
{p_{\mathrm{neg}}(a_d\mid q,m_d)}.
\end{align*}
Then
\begin{equation*}
s^*(q,d)
=
\phi(q,M_d)+\psi(q,M_d,A_d)+C(q).
\end{equation*}
If
\[
p_{\mathrm{pos}}(A_d\mid q,M_d)
\neq
p_{\mathrm{neg}}(A_d\mid q,M_d)
\]
on a set of positive measure, then the likelihood ratio
\[
\frac{p_{\mathrm{pos}}(A_d\mid q,M_d)}
{p_{\mathrm{neg}}(A_d\mid q,M_d)}
\]
is not almost surely constant on that set. Hence $\psi(q,M_d,A_d)$ is non-constant as a function of $A_d$ there. Consequently, $s^*(q,d)$ cannot in general be represented as a function of $(q,M_d)$ alone. This establishes that the optimal scorer contains a conditional artifact contribution under the stated condition.
\end{proof}

\subsection{Proof of Proposition~\ref{prop:sem-art}}

\begin{proof}
Let $g(m,a)$ denote the document embedding as a function of semantic and artifact features:
$g(m,a)=h_d(d)$ for documents with $(M_d,A_d)=(m,a).$
Assume $g$ is differentiable with respect to $a$ in a neighborhood of a reference artifact value $a_0$. A first-order Taylor expansion around $a_0$ gives
\begin{align*}
g(m,a)
=
g(m,a_0)
+
J_a(m,a_0)(a-a_0)
+
r(m,a),
\end{align*}
where
\[
J_a(m,a_0)
=
\left[
\frac{\partial g(m,a)}{\partial a}
\right]_{a=a_0}
\]
and
\[
\|r(m,a)\|=o(\|a-a_0\|)
\quad
\text{as}
\quad
\|a-a_0\|\to 0.
\]
Define
\begin{align*}
&h_{\mathrm{sem}}(m)=g(m,a_0), \\
&h_{\mathrm{art}}(a;m)=J_a(m,a_0)(a-a_0).
\end{align*}
Then
\[
h_d(d)
=
h_{\mathrm{sem}}(M_d)
+
h_{\mathrm{art}}(A_d;M_d)
+
r(M_d,A_d).
\]
For a dot-product retriever,
\begin{align*}
s_\theta(q,d)
&=
\langle h_q(q),h_d(d)\rangle
\\
&=
\langle h_q(q),h_{\mathrm{sem}}(M_d)\rangle \\
&\quad
+
\langle h_q(q),h_{\mathrm{art}}(A_d;M_d)\rangle
\\
&\quad
+
\langle h_q(q),r(M_d,A_d)\rangle.
\end{align*}
Let
\[
\varepsilon(q,M_d,A_d)
=
\langle h_q(q),r(M_d,A_d)\rangle.
\]
Since $\|r(M_d,A_d)\|=o(\|A_d-a_0\|)$, we have
\[
\varepsilon(q,M_d,A_d)
=
o(\|A_d-a_0\|)
\]
for fixed $h_q(q)$.
Thus,
\begin{align*}
s_\theta(q,d)
&=
\underbrace{\langle h_q(q),h_{\mathrm{sem}}(M_d)\rangle}_{\text{semantic reference}} \\
&\quad +
\underbrace{\langle h_q(q),h_{\mathrm{art}}(A_d;M_d)\rangle}_{\text{artifact perturbation}}
+
\varepsilon(q,M_d,A_d).
\end{align*}

If $J_a(M_d,a_0)$ varies weakly with $M_d$, or if its variation is dominated by a shared low-dimensional component, then the artifact perturbation can be approximated by a shared direction or subspace in the embedding space. This gives the idealized linear view used to motivate the empirical embedding-direction analysis.
\end{proof}

\paragraph{Remark.}
The embedding-space argument is local and first-order. It does not require semantic and artifact information to be globally separable in the learned representation. Rather, it shows that, around a reference artifact value, changes in artifact features can induce approximately linear perturbations in the document embedding. This motivates examining whether LLM--human and positive--negative contrasts exhibit stable directions in the learned embedding space.
\section{Additional Linguistic Analyses}
\label{app:linguistic_appendix}

This appendix supplements Section~\ref{subsec:rq2_linguistic} with effect sizes, full-benchmark results, Qwen3-8B validation, and artifact--score analysis.

\subsection{Effect-Size Analyses}
\label{app:linguistic_effectsize}
We quantify the magnitude of linguistic differences using standard effect-size measures (Hedges’ $g$ for mean differences) and report associated significance levels. 
These statistics complement the significance tests in the main paper by showing not only whether differences are significant but also their practical magnitude. 
Table~\ref{tab:effect_size} summarizes results on MS~MARCO for two contrasts: (i) positives vs.\ the unannotated pool, and (ii) LLM-generated vs.\ human-written passages. 

\begin{table}[ht]
\centering
\caption{Effect sizes (Hedges’ $g$) and significance for linguistic feature comparisons on MS~MARCO.
Positive values indicate higher scores for the first group. $p$-values smaller than numerical precision are reported as $p<10^{-15}$.}
\label{tab:effect_size}
\resizebox{\columnwidth}{!}{%
\begin{tabular}{lccc}
\toprule
Comparison & PPL ($g$) & IDF ($g$) & $p$-value \\
\midrule
Positives vs.\ Unannotated & $-0.214$ & $+0.047$ & $<10^{-15}$ \\
LLM vs.\ Human             & $-0.274$ & $+0.145$ & $<10^{-15}$ \\
\bottomrule
\end{tabular}%
}
\end{table}

We observe that both comparisons yield highly significant differences despite modest effect sizes. 
For perplexity (PPL), positives are more fluent than the unannotated pool ($g=-0.214$), and LLM passages are even more fluent than human passages ($g=-0.274$). 
For IDF, the effects are smaller ($g=0.047$ and $0.145$ respectively) but consistently positive, indicating that both positives and LLM rewrites exhibit slightly greater lexical specificity.
Taken together, these results show that supervision and source type both introduce systematic, statistically robust shifts in linguistic features, even if the magnitudes are moderate. 

\subsection{Positives vs. Negatives on Additional Datasets}
\label{app:linguistic_posneg}

To assess whether the imbalance between annotated positives and negatives generalizes beyond MS~MARCO, we extend the perplexity analysis to other datasets in Cocktail (Figure~\ref{fig:pos_neg_ppl_distribution}).
For datasets that share the same corpus (e.g., MS~MARCO and DL19/20), we report results only once. For NFCorpus and HotpotQA, all passages are annotated with relevance labels, so no negative pool exists and only positives are shown. Across the remaining datasets, positives consistently exhibit lower perplexity than negatives, mirroring the trend in MS~MARCO. This indicates that stylistic disparities between positives and negatives are not dataset-specific idiosyncrasies but a systematic property of retrieval supervision. As discussed in the main text, positives are often drawn from edited, high-quality sources intended to serve as good answers, whereas negatives derive from more heterogeneous and less polished text.

\begin{figure*}[ht]
  \centering
  \includegraphics[width=0.95\linewidth]{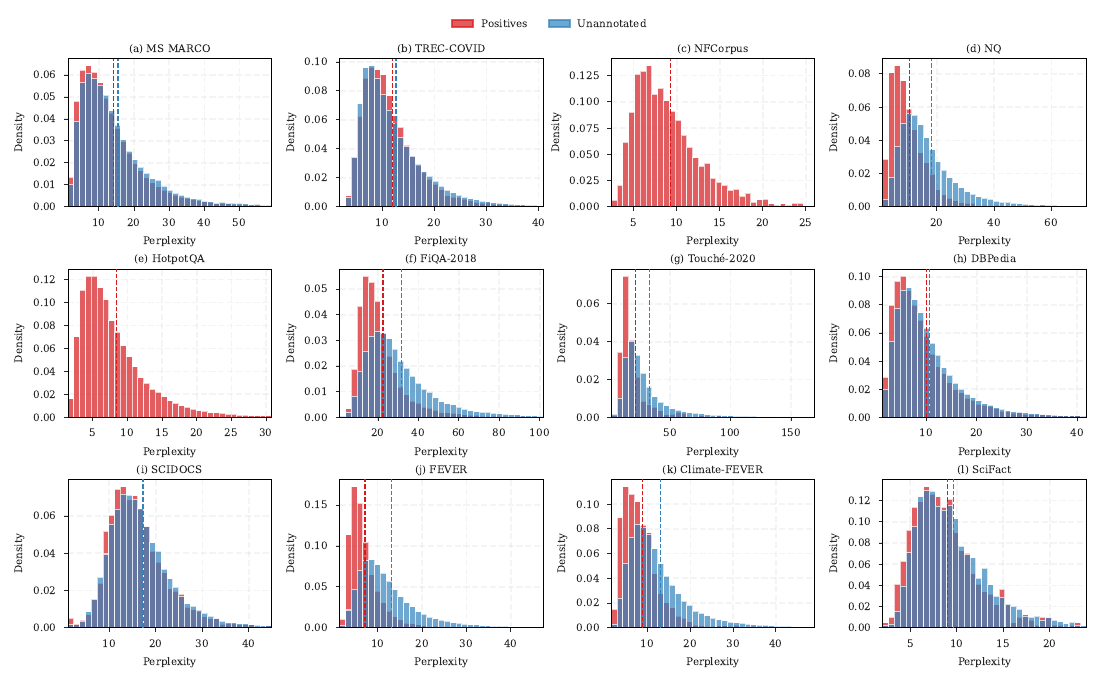}
  \caption{
    Perplexity distributions of positives versus negatives across retrieval datasets in Cocktail. For MS~MARCO and DL19/20, results are reported once due to corpus overlap. For NFCorpus and HotpotQA, all passages are annotated as relevant, so only positive distributions are shown. 
    }
  \label{fig:pos_neg_ppl_distribution}
\end{figure*}

\subsection{LLM vs. Human across Additional Datasets}
\label{app:linguistic_otherdatasets}
To ensure that the findings generalize beyond MS~MARCO, we replicate the analysis on the other datasets in Cocktail. 
Figure~\ref{fig:ppl_otherdatasets} reports perplexity distributions, and Figure~\ref{fig:idf_otherdatasets} reports IDF distributions, comparing LLM-generated versus human-written passages. 

Consistent with the MS~MARCO case, LLM-generated passages consistently exhibit lower perplexity and slightly higher IDF than their human-written counterparts. The PPL differences are stable and consistent across all datasets, while the IDF differences are more modest in magnitude but follow the same direction throughout. These results confirm that source-based stylistic artifacts are systematic and broadly consistent across domains.

\begin{figure*}[t]
    \centering
    \includegraphics[width=\linewidth]{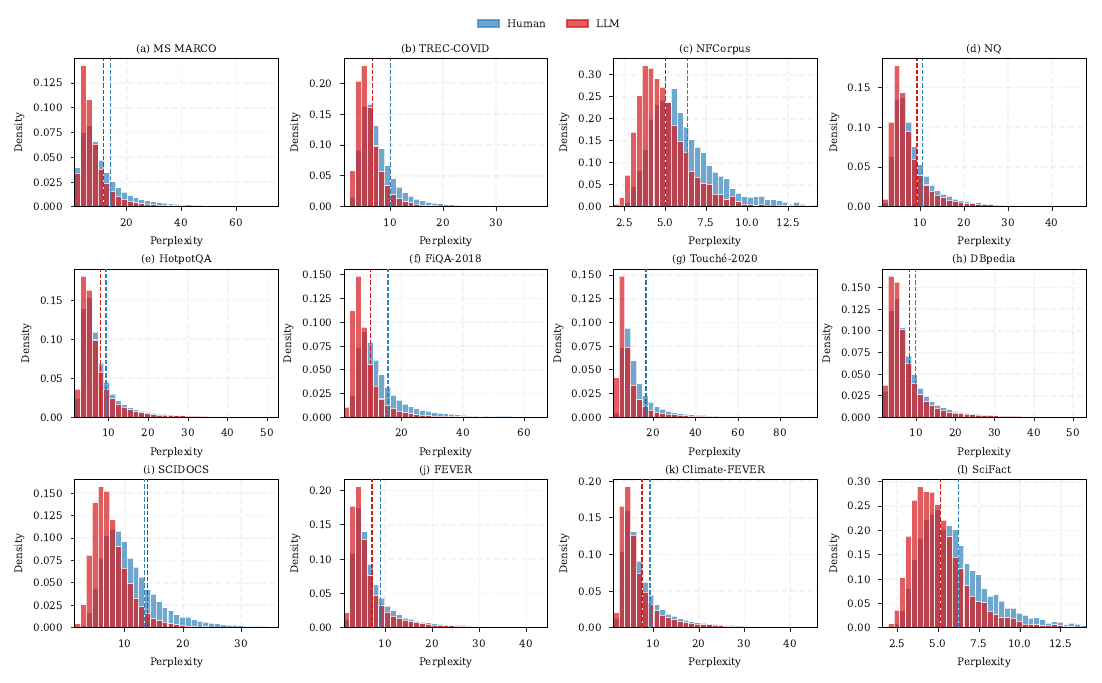}
    \caption{Perplexity (PPL) distributions of LLM-generated vs. human-written passages across additional datasets. 
    Red = LLM, Blue = Human. LLM passages consistently exhibit lower perplexity, indicating higher fluency.}
    \label{fig:ppl_otherdatasets}
\end{figure*}

\begin{figure*}[t]
    \centering
    \includegraphics[width=\linewidth]{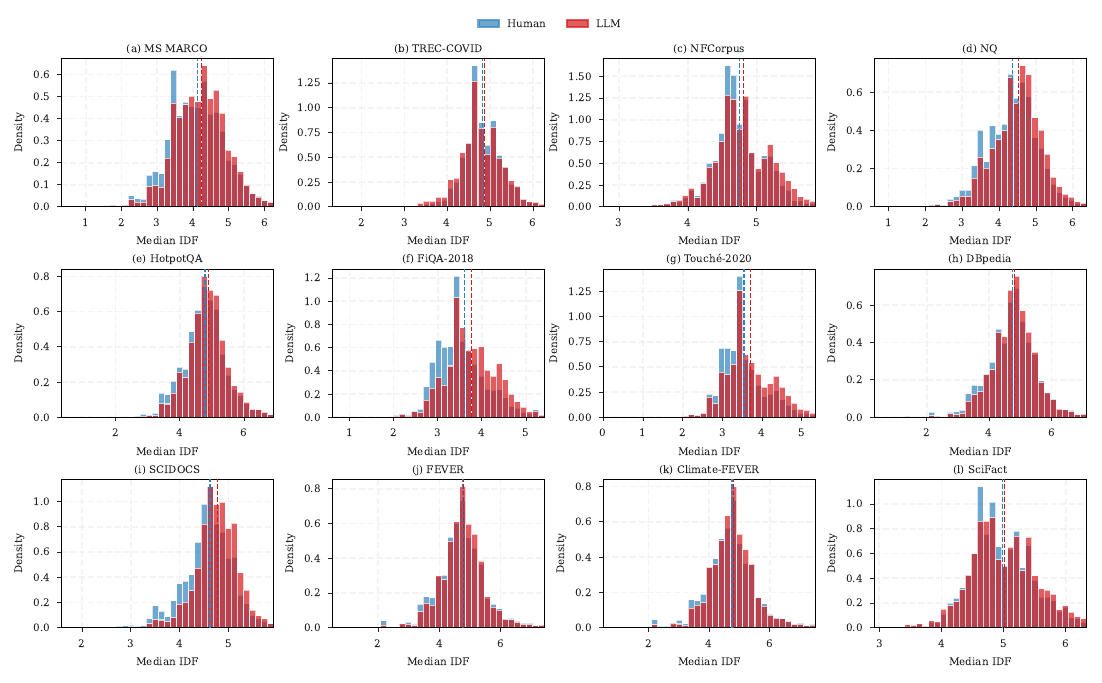}
    \caption{Median IDF distributions of LLM-generated vs. human-written passages across additional datasets. 
    Red = LLM, Blue = Human. LLM passages generally exhibit higher IDF, though the gap varies across datasets.}
    \label{fig:idf_otherdatasets}
\end{figure*}

\subsection{Cross-Family Perplexity Validation}
\label{app:qwen_ppl}

Using Qwen3-8B~\citep{qwen3technicalreport} as an independent scorer, we recompute passage PPL across the same corpora. Table~\ref{tab:qwen_ppl} shows lower PPL for LLM-generated passages in all 12 unique corpora, matching the Llama-3-8B-Instruct result.

\begin{table}[t]
\centering
\small
\caption{Mean passage PPL computed with Qwen3-8B. MS~MARCO, DL19, and DL20 share the same corpus.}
\label{tab:qwen_ppl}
\renewcommand{\arraystretch}{1.0}
\setlength{\tabcolsep}{5pt}
\resizebox{\columnwidth}{!}{%
\begin{tabular}{lrr}
\toprule
Corpus & Human & Llama-2 generated \\
\midrule
MS~MARCO / DL19 / DL20 & 15.04 & 11.14 \\
TREC-COVID             & 10.60 &  9.39 \\
NFCorpus               &  7.56 &  6.25 \\
NQ                     & 18.98 & 15.77 \\
HotpotQA               & 10.57 &  9.59 \\
FiQA-2018              & 31.09 & 17.79 \\
Touché-2020            & 32.81 & 22.50 \\
DBpedia                & 12.21 & 10.42 \\
SCIDOCS                & 14.45 & 10.31 \\
FEVER                  & 13.53 &  9.98 \\
Climate-FEVER          & 13.96 & 10.35 \\
SciFact                &  7.75 &  6.72 \\
\bottomrule
\end{tabular}
}
\end{table}

\subsection{Artifact--Score Correlations}
\label{app:artifact_score_corr}

We directly measure how the linguistic artifacts relate to retriever scores over approximately 57.6 million query--document pairs. For each retriever, document scores are z-normalized within each query before computing Pearson correlations with document-level PPL and IDF. This preserves rankings while aligning query-specific score scales.

Table~\ref{tab:artifact_score_corr} shows a consistent relationship across all five relevance-supervised retrievers: lower PPL and higher IDF are associated with higher retrieval scores. The coefficient signs are stable across retrievers, and all associations are statistically significant. Systematic score shifts of this form can alter top-ranked results when candidate scores are close.
\section{Cosine Similarity Between Random High-Dimensional Vectors}
\label{app:cosine_null}

We derive the null distribution of cosine similarities between independent random vectors, which serves as the statistical baseline for our embedding-space analyses. Let $x,y \in \mathbb{R}^m$ be isotropic random vectors. Normalizing to the unit sphere ($\hat x = x/\|x\|$, $\hat y = y/\|y\|$) yields $\hat x,\hat y \sim \mathrm{Unif}(\mathbb{S}^{m-1})$, and their cosine similarity is
\[
Z = \langle \hat x, \hat y \rangle \in [-1,1].
\]
By rotational invariance, $Z$ follows a Beta-type density~\citep{vershynin2018high}:
\[
f_{Z}(z) \;=\;
\frac{\Gamma(\tfrac m2)}{\sqrt{\pi}\,\Gamma(\tfrac{m-1}{2})}
(1-z^2)^{\frac{m-3}{2}}, \quad z \in [-1,1],
\]
which is symmetric around zero. Equivalently, the tail probability can be expressed via the regularized incomplete Beta function:
\[
\Pr(|Z| > t) \;=\; \mathrm{I}_{\,1-t^2}\!\Big(\tfrac{m-1}{2},\,\tfrac12\Big).
\]

By symmetry, $\mathbb{E}[Z]=0$. Since each coordinate of a uniform unit vector has variance $1/m$, the variance of $Z$ is
\[
\mathrm{Var}(Z) \;=\; \frac{1}{m}.
\]
For large $m$, the density concentrates sharply at zero. Expanding $\log(1-z^2)\approx -z^2$ near the origin gives the Gaussian approximation
\[
Z \;\approx\; \mathcal{N}\!\left(0,\,\tfrac{1}{m}\right).
\]

In dimension $m=768$, the standard deviation is $\sigma = 1/\sqrt{m} \approx 0.0361$, so that $3\sigma \approx 0.108$. Under the normal approximation,
\[
\Pr(|Z| > 3\sigma) \approx 0.27\%,
\]
which closely matches the exact Beta distribution. Thus, over $99.7\%$ of random pairs fall within $\pm 3\sigma$, validating the use of this threshold as a significance criterion in high-dimensional embedding spaces. Figure~\ref{fig:cosine_null_distribution} illustrates this concentration.

\begin{figure}[ht]
    \centering
    \includegraphics[width=\linewidth]{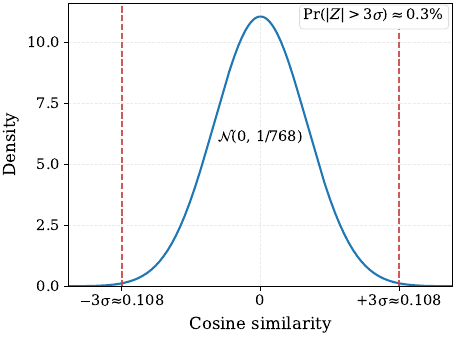}
    \caption{Null distribution of cosine similarity between random vectors in $m=768$ dimensions, approximated by $\mathcal{N}(0,1/m)$. Over $99.7\%$ of values lie within $\pm 3\sigma \approx 0.108$, supporting its use as a significance criterion.}
    \label{fig:cosine_null_distribution}
\end{figure}
\section{Additional Embedding Analyses}
\label{app:embedding_full}

This appendix provides three complementary sets of embedding-space analyses. First, we extend the
representative results in the main text to the full set of unique corpora in the Cocktail
benchmark, using the DRAGON retriever as a representative model. Second, because dense retriever
embeddings are not necessarily isotropic, we report an empirical paired sign-flip validation that
uses the observed paired displacements to construct a more realistic null distribution. Third, we
examine directional heterogeneity across individual passage pairs and motivate a multi-dimensional
representation of source-related structure.

\subsection{Additional Embedding Analyses on the Full Benchmark}
\label{app:embedding_full_visual}

Our experiments use 14 datasets from the Cocktail benchmark. Since three of them
(MS~MARCO, DL19, and DL20) share the same underlying corpus, we report visualization results at
the corpus level, resulting in 12 unique corpora. These figures complement the representative
results shown in the main text and report:
(1) within-dataset displacement consistency (Figure~\ref{fig:app_within}),
(2) cross-dataset similarity of mean displacement directions (Figure~\ref{fig:app_cross}), and
(3) alignment between LLM--human and supervision-induced directions
(Figure~\ref{fig:app_alignment}).

\paragraph{Random-vector reference.}
The dashed $3\sigma$ reference lines in Figures~\ref{fig:app_within} and
\ref{fig:app_alignment} are based on the cosine similarities of random vectors in the same
embedding dimensionality. We use this as a descriptive geometric reference: in high-dimensional
space, independent isotropic random vectors are nearly orthogonal, so values outside the
$3\sigma$ range indicate unusually strong alignment compared with idealized random directions.
However, dense retriever embeddings are not isotropic and may contain common components,
dataset-level structure, and anisotropy. Therefore, we complement this reference with the
empirical paired sign-flip validation in Appendix~\ref{app:signflip_validation}.

\begin{figure*}[t]
  \centering
  \includegraphics[width=0.9\linewidth]{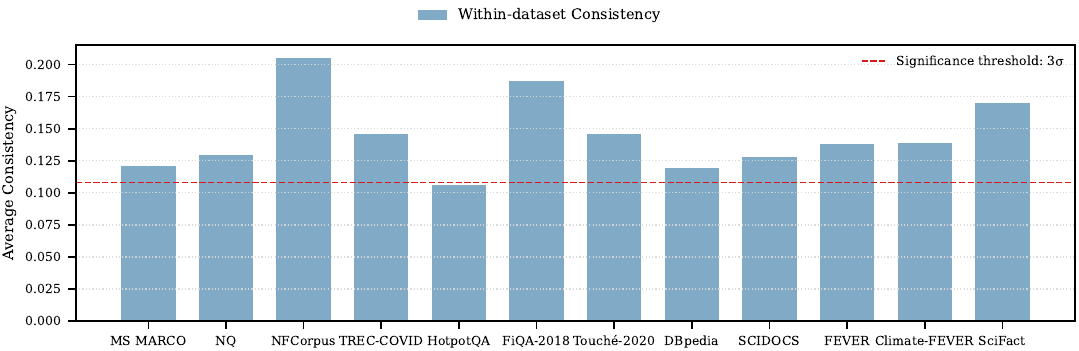}
  \caption{
  Within-dataset consistency of LLM--Human displacements. Bars show average pairwise cosine
  similarity among displacement vectors $\delta_i^{\mathrm{LH}}$ within each dataset, relative
  to the $3\sigma$ random-vector reference. Most datasets exceed the reference range, with a
  few exceptions near or below it.
  }
  \label{fig:app_within}
\end{figure*}

\begin{figure*}[t]
  \centering
  \includegraphics[width=0.9\linewidth]{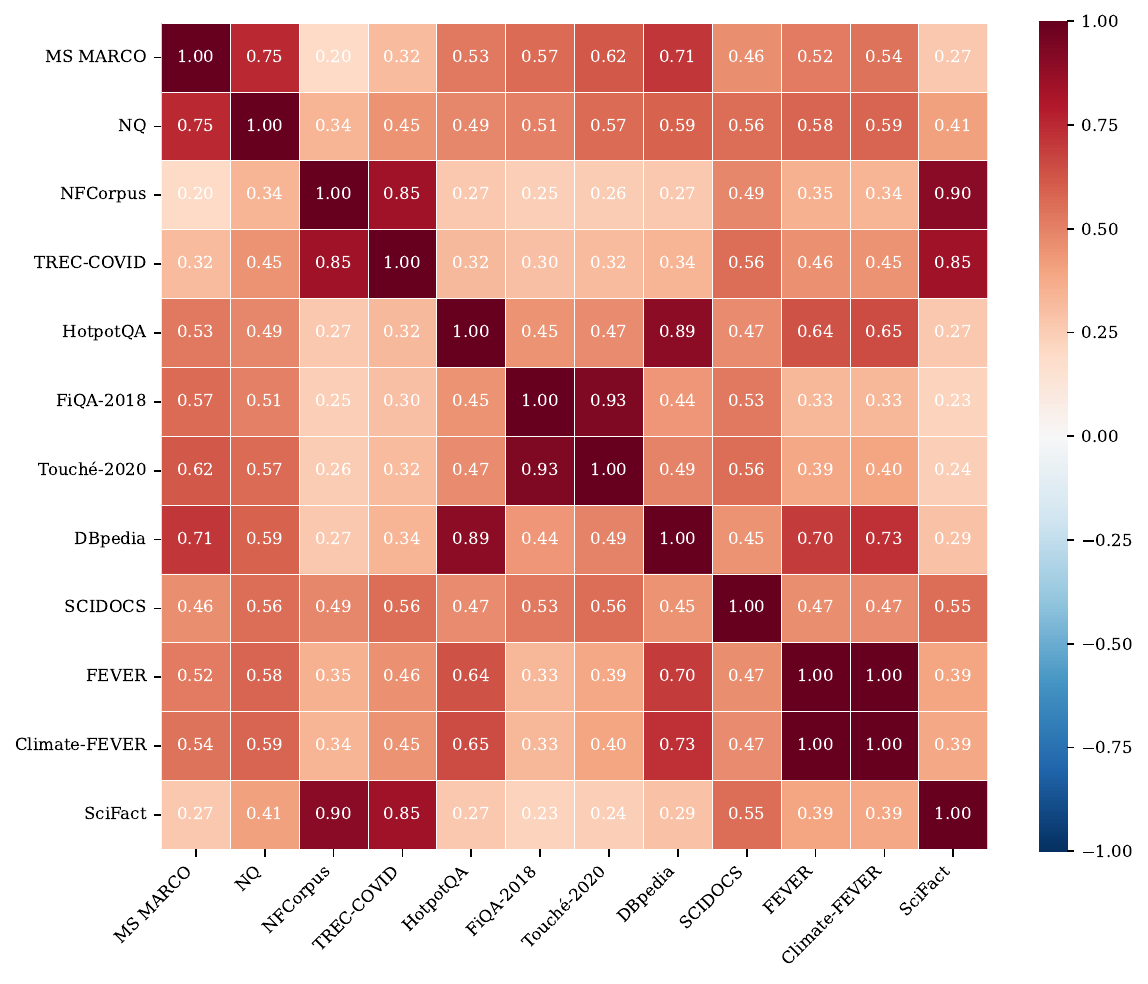}
  \caption{
  Cross-dataset similarity of mean LLM--Human displacement directions. Values denote cosine
  similarity between dataset-level means $\overline{\delta}_{\mathrm{LH},D}$. Darker cells
  indicate stronger alignment, revealing consistent source-related directions across corpora.
  }
  \label{fig:app_cross}
\end{figure*}

\begin{figure*}[t]
  \centering
  \includegraphics[width=0.9\linewidth]{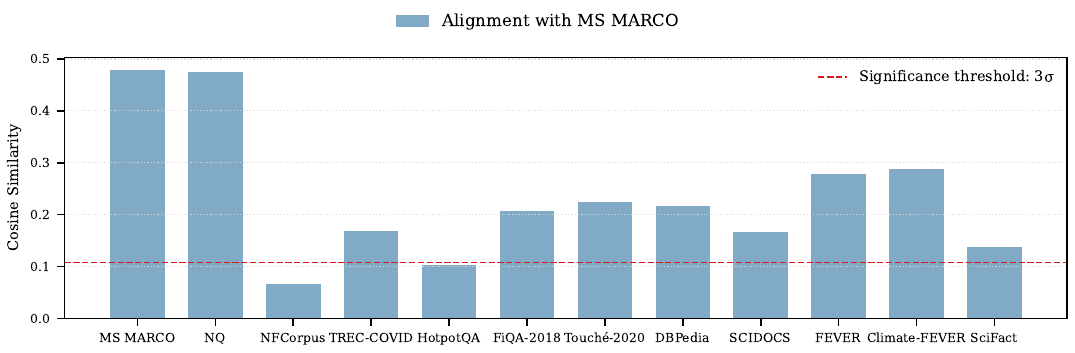}
  \caption{
  Cosine similarity between the LLM--Human displacement direction and the MS~MARCO
  positive--negative contrast across datasets. The red dashed line marks the $3\sigma$
  random-vector reference. Most datasets show strong alignment beyond the reference range, with a
  few cases near or below it.
  }
  \label{fig:app_alignment}
\end{figure*}

Overall, these full-benchmark visualizations extend the main-text findings to the full set of
unique corpora. The majority of corpora follow the same trends as reported in the main text,
while a small number exhibit weaker effects, which we treat as exceptions rather than
contradictions.

\subsection{Empirical Paired Sign-Flip Validation}
\label{app:signflip_validation}

The random-vector reference above provides an idealized geometric baseline, but it does not fully
capture the geometry of real dense retriever embeddings. Retriever embeddings may be anisotropic
and may contain common components, dataset-level structure, or other non-random regularities. We
therefore complement the random-vector reference with a paired sign-flip permutation test.

For each human--LLM passage pair, we define the LLM--human displacement as
\[
\delta_i^{\mathrm{LH}}
=
h_d(d_i^{\mathrm{LLM}})
-
h_d(d_i^{\mathrm{Human}}).
\]
In each sign-flip permutation, we randomly multiply each paired displacement by $+1$ or $-1$:
\[
\delta_i^{\pi}
=
z_i \delta_i^{\mathrm{LH}},
\qquad
z_i \sim \mathrm{Unif}\{-1,+1\}.
\]
This procedure preserves the magnitude of each displacement, the semantic pairing between the
human-written passage and its LLM-generated counterpart, the dataset composition, and the
empirical geometry of the retriever embedding space. At the same time, it removes any systematic
orientation from human-written to LLM-generated text. The resulting sign-randomized distribution
therefore serves as an empirical null for testing whether the observed directionality is stronger
than expected under the same embedding geometry.

For a statistic $T$, let $T_{\mathrm{obs}}$ denote the value computed from the original
displacements, and let $T_1^{\pi},\ldots,T_B^{\pi}$ denote the values computed from $B$
sign-flip permutations. We use the one-sided empirical permutation $p$-value
\[
p
=
\frac{
1 + \sum_{b=1}^{B}
\mathbf{1}
\left[
T_b^{\pi} \geq T_{\mathrm{obs}}
\right]
}{
B+1
}.
\]
The additive correction prevents zero $p$-values with a finite number of Monte Carlo
permutations. We use $B=1000$ permutations, so the smallest attainable corrected permutation
$p$-value is $1/(B+1)=0.000999$.

We apply this validation to the three embedding-space quantities analyzed in the main text and in
Appendix~\ref{app:embedding_full_visual}.

\paragraph{Within-dataset LLM--human consistency.}
For each dataset $D$, we first compute the dataset-level mean LLM--human displacement
\[
\overline{\delta}_{\mathrm{LH},D}
=
\frac{1}{N_D}
\sum_{i\in D}
\delta_i^{\mathrm{LH}}.
\]
We then measure within-dataset consistency as the average alignment between each pair-level
LLM--human displacement and this dataset-level mean displacement:
\[
T^{\mathrm{within}}_{\mathrm{obs}}(D)
=
\frac{1}{N_D}
\sum_{i\in D}
\cos
\left(
\delta_i^{\mathrm{LH}},
\overline{\delta}_{\mathrm{LH},D}
\right).
\]
This statistic matches the within-dataset consistency analyzed in
Appendix~\ref{app:embedding_full_visual}: it asks whether individual LLM--human displacements
within the same dataset consistently point toward the dataset-level source direction.

For each sign-flip permutation, we replace each displacement with
$\delta_i^{\pi}=z_i\delta_i^{\mathrm{LH}}$ and recompute the permuted dataset-level mean
direction
\[
\overline{\delta}^{\pi}_{\mathrm{LH},D}
=
\frac{1}{N_D}
\sum_{i\in D}
z_i\delta_i^{\mathrm{LH}}.
\]
The corresponding null statistic is
\[
T^{\mathrm{within}}_{\pi}(D)
=
\frac{1}{N_D}
\sum_{i\in D}
\cos
\left(
z_i\delta_i^{\mathrm{LH}},
\overline{\delta}^{\pi}_{\mathrm{LH},D}
\right).
\]
This tests whether pair-level LLM--human displacements are more aligned with their dataset-level
mean source direction than expected after randomizing the human-to-LLM orientation within each
paired displacement.

\paragraph{Cross-dataset LLM--human consistency.}
For each dataset $D$, we compute the dataset-level mean displacement
\[
\overline{\delta}_{\mathrm{LH},D}
=
\frac{1}{N_D}
\sum_{i\in D}
\delta_i^{\mathrm{LH}}.
\]
The observed cross-dataset consistency is the average cosine similarity between dataset-level
mean directions:
\[
T^{\mathrm{cross}}_{\mathrm{obs}}
=
\frac{2}{K(K-1)}
\sum_{a<b}
\cos
\left(
\overline{\delta}_{\mathrm{LH},D_a},
\overline{\delta}_{\mathrm{LH},D_b}
\right),
\]
where $K$ is the number of dataset-level collections. For each permutation, we independently
sign-flip the paired LLM--human displacements within each dataset, recompute the dataset-level
mean directions,
\[
\overline{\delta}^{\pi}_{\mathrm{LH},D}
=
\frac{1}{N_D}
\sum_{i\in D}
z_i\delta_i^{\mathrm{LH}},
\]
and then recompute the same cross-dataset statistic:
\[
T^{\mathrm{cross}}_{\pi}
=
\frac{2}{K(K-1)}
\sum_{a<b}
\cos
\left(
\overline{\delta}^{\pi}_{\mathrm{LH},D_a},
\overline{\delta}^{\pi}_{\mathrm{LH},D_b}
\right).
\]
This tests whether the observed cross-dataset alignment remains stronger than expected after
randomizing the human-to-LLM orientation within each dataset.

\paragraph{Positive--negative and LLM--human alignment.}
We also validate the alignment between the supervision-induced positive--negative direction and
the LLM--human source direction. The positive--negative direction is estimated from MS~MARCO as
\[
\overline{\delta}_{\mathrm{PN}}
=
\mathbb{E}
\left[
h_d(d^+) - h_d(d^-)
\right],
\]
where $d^+$ is an annotated positive passage and $d^-$ is a BM25-retrieved negative passage. We
treat $\overline{\delta}_{\mathrm{PN}}$ as a fixed supervision-induced reference direction. For
each dataset $D$, the observed PN--LH alignment is
\[
T^{\mathrm{PN\text{-}LH}}_{\mathrm{obs}}(D)
=
\cos
\left(
\overline{\delta}_{\mathrm{PN}},
\overline{\delta}_{\mathrm{LH},D}
\right).
\]
The sign-flip null keeps $\overline{\delta}_{\mathrm{PN}}$ fixed and applies sign flips only to
the LLM--human displacements:
\[
T^{\mathrm{PN\text{-}LH}}_{\pi}(D)
=
\cos
\left(
\overline{\delta}_{\mathrm{PN}},
\overline{\delta}^{\pi}_{\mathrm{LH},D}
\right).
\]
This tests whether the observed LLM--human direction is more strongly aligned with the
positive--negative supervision direction than expected if the human-to-LLM orientation were
randomized within each pair.

\paragraph{Aggregating permutation results.}
For dataset-level quantities, namely within-dataset LH consistency and PN--LH alignment, we report
an overall statistic by macro-averaging over datasets. Specifically, for a statistic $T(D)$
defined on dataset $D$, the observed overall value is
\[
T_{\mathrm{obs}}^{\mathrm{overall}}
=
\frac{1}{K}
\sum_{D=1}^{K}
T_{\mathrm{obs}}(D),
\]
where $K$ is the number of dataset-level collections. The corresponding overall null distribution
is constructed by applying the same macro-average within each sign-flip permutation:
\[
T_{\pi_b}^{\mathrm{overall}}
=
\frac{1}{K}
\sum_{D=1}^{K}
T_{\pi_b}(D).
\]
The null mean and standard deviation in Table~\ref{tab:signflip_summary} are computed from
$\{T_{\pi_b}^{\mathrm{overall}}\}_{b=1}^{B}$, and the reported $p$-value compares
$T_{\mathrm{obs}}^{\mathrm{overall}}$ against this overall null distribution. We do not average
per-dataset $p$-values.

For cross-dataset LH consistency, the statistic is already defined globally as the average
cosine similarity over all pairs of dataset-level mean directions. Therefore, its observed value
and null distribution are computed directly from the cross-dataset statistic for each sign-flip
permutation. All sign-flip validations are computed over the 14 dataset-level collections used in
the benchmark.

Table~\ref{tab:signflip_summary} summarizes the paired sign-flip validation using DRAGON. All
three observed statistics are substantially larger than their corresponding sign-flip null
distributions. With $B=1000$ permutations, the smallest attainable permutation $p$-value under the
additive correction is $1/(B+1)=0.000999$, reported as $0.0010$ in the table. Thus, the
within-dataset LH consistency, cross-dataset LH consistency, and PN--LH alignment are all
significant under the paired sign-flip test. These results support the conclusion that the
observed embedding-space patterns are not artifacts of the idealized random-vector reference or
of the marginal geometry of the embedding space.

\begin{table}[t]
\centering
\small
\setlength{\tabcolsep}{4pt}
\resizebox{\linewidth}{!}{
\begin{tabular}{lrrrr}
\toprule
Metric & Obs. & Null $\mu$ & Null $\sigma$ & $p$ \\
\midrule
Within-dataset LH & 0.1436 & 0.0044 & 0.0001 & 0.0010 \\
Cross-dataset LH & 0.5089 & 0.0001 & 0.0053 & 0.0010 \\
 PN--LH alignment & 0.2692 & 0.0015 & 0.0240 & 0.0010 \\
\bottomrule
\end{tabular}
}
\caption{
Summary of paired sign-flip validation using DRAGON over the 14 dataset-level collections.
Obs. denotes the observed statistic; Null $\mu$ and Null $\sigma$ denote the mean and standard
deviation of the corresponding overall sign-flip null distribution. For dataset-level metrics,
the overall statistic is obtained by macro-averaging over datasets within each permutation.
With $B=1000$ permutations, the minimum attainable corrected permutation $p$-value is
$1/(B+1)=0.000999$.
}
\label{tab:signflip_summary}
\end{table}

\subsection{Directional Heterogeneity and Multi-Dimensional Structure}
\label{app:embedding_heterogeneity}

The corpus-level mean direction captures the source-related component shared across many passage pairs, while individual human--LLM transformations can exhibit richer local geometry. For the FEVER passage \textit{Infinity\_Watch}, the individual human-to-LLM displacement has cosine similarity $-0.22$ with the corpus-level direction, compared with a corpus average of $0.15$. This locally reversed displacement illustrates variation that a single global direction cannot represent.

The same geometry naturally motivates extending the source direction to a multi-dimensional artifact subspace. Multiple dimensions can capture recurring source-related transformations beyond the dominant mean displacement while preserving the cross-document structure established by our consistency analyses. Relevance-preserving subspace learning offers stronger source-aware control.
\section{Mitigation Effectiveness, Stability, and Generalization}
\label{app:rq3_results}

This appendix details training-time effectiveness and stability, full inference-time results, the CDC comparison, and Qwen-generated corpus results for Section~\ref{sec:rq3}.

\subsection{Training-time Interventions: Retrieval Effectiveness}
Table~\ref{tab:rq3_train_performance_app} reports retrieval effectiveness (NDCG@5) for the three negative sampling strategies (\emph{in-batch only}, \emph{standard}, and \emph{hard-neg only}) across all datasets. Overall, settings that include mined hard negatives achieve higher retrieval performance, while using only in-batch negatives leads to lower effectiveness on most datasets. This trend is consistent with widely noted observations in the dense retrieval community that mined hard negatives are essential for strong retrieval effectiveness.

\begin{table}[ht]
    \centering
    \caption{NDCG@5 results on 14 datasets under different negative sampling strategies. 
    The \emph{standard} strategy combines in-batch and mined hard negatives, while the other two use only one type of negative.}
    \label{tab:rq3_train_performance_app}
    \resizebox{\columnwidth}{!}{%
    \begin{tabular}{lccc}
    \toprule
    \textbf{Dataset} & In-batch only & Standard & Hard-neg only \\
    \midrule
    MS~MARCO      & 0.629 & 0.629 & 0.623 \\
    DL19          & 0.640 & 0.706 & 0.728 \\
    DL20          & 0.642 & 0.701 & 0.719 \\
    TREC-COVID    & 0.571 & 0.611 & 0.568 \\
    NFCorpus      & 0.303 & 0.287 & 0.278 \\
    NQ            & 0.652 & 0.670 & 0.666 \\
    HotpotQA      & 0.570 & 0.579 & 0.579 \\
    FiQA-2018     & 0.209 & 0.218 & 0.216 \\
    Touché-2020   & 0.350 & 0.418 & 0.411 \\ 
    DBpedia       & 0.428 & 0.436 & 0.437 \\
    SCIDOCS       & 0.096 & 0.086 & 0.086 \\
    FEVER         & 0.850 & 0.842 & 0.829 \\
    Climate-FEVER & 0.280 & 0.271 & 0.241 \\
    SciFact       & 0.435 & 0.452 & 0.443 \\
    \midrule
    \textbf{Average}  & 0.475 & 0.493 & 0.487 \\
    \bottomrule
    \end{tabular}%
    }
\end{table}

\subsection{Multi-Seed Stability of Negative Sampling}

We repeat the three original negative-sampling conditions with seeds 42, 100, and 2026. Table~\ref{tab:rq3_multiseed} reports mean absolute $\Delta$NDSR@5 across the 14 datasets, rather than the signed averages in Table~\ref{tab:rq3_train}. Every seed preserves the same ordering: in-batch only produces the weakest source preference, followed by standard and hard-neg-only sampling.

\begin{table}[t]
\centering
\small
\caption{Multi-seed stability of the RQ3 negative-sampling experiment. Values are mean absolute $\Delta$NDSR@5 across 14 datasets.}
\label{tab:rq3_multiseed}
\renewcommand{\arraystretch}{1.0}
\setlength{\tabcolsep}{5pt}
\begin{tabular}{lccc}
\toprule
Seed & In-batch only & Standard & Hard-neg only \\
\midrule
42   & 0.050 & 0.099 & 0.100 \\
100  & 0.041 & 0.088 & 0.090 \\
2026 & 0.053 & 0.085 & 0.091 \\
\bottomrule
\end{tabular}
\end{table}

\subsection{Inference-time Interventions}
Table~\ref{tab:rq3_debias_app} reports $\Delta$NDSR@5 across all 14 datasets, extending the five-dataset summary in the main text. Tables~\ref{tab:rq3_cdc_app} and~\ref{tab:rq3_qwen_app} provide the direct CDC comparison and Qwen-generator evaluation.

\begin{table*}[t]
\centering
\caption{$\Delta$NDSR@5 before and after inference-time debiasing across 14 datasets and 5 relevance-supervised retrievers, using leave-one-dataset-out calibration to estimate the debiasing direction. 
Positive values indicate human preference, while negative values indicate LLM preference. 
In the Average row, the first line reports mean $\Delta$NDSR@5, and the second line reports the remaining $|\Delta\mathrm{NDSR}@5|$ relative to the original value. 
Darker shading indicates stronger source bias.}
\label{tab:rq3_debias_app}

\resizebox{\textwidth}{!}{
\begin{minipage}{0.92\textwidth}
\renewcommand{\arraystretch}{1.0}
\setlength{\tabcolsep}{3.5pt}
\resizebox{\linewidth}{!}{
\begin{tabular}{l | cc | cc | cc | cc | cc}
\toprule
\multirow{2}{*}{Dataset (↓)} 
  & \multicolumn{2}{c}{ANCE} 
  & \multicolumn{2}{c}{coCondenser} 
  & \multicolumn{2}{c}{DRAGON} 
  & \multicolumn{2}{c}{RetroMAE} 
  & \multicolumn{2}{c}{TAS-B} \\
\cmidrule(lr){2-3} \cmidrule(lr){4-5} \cmidrule(lr){6-7} \cmidrule(lr){8-9} \cmidrule(lr){10-11}
 & Original & Debias & Original & Debias & Original & Debias & Original & Debias & Original & Debias \\
\midrule
MS~MARCO          & -0.042 & -0.004 & -0.020 & -0.002 & -0.083 & -0.035 & -0.083 & -0.053 & -0.121 & -0.064 \\
DL19              & -0.073 & -0.029 & -0.072 & -0.041 & -0.233 & -0.161 & -0.186 & -0.131 & -0.224 & -0.176 \\
DL20              & -0.034 &  0.021 & -0.079 & -0.068 & -0.121 & -0.069 & -0.088 & -0.054 & -0.072 & -0.065 \\
TREC-COVID        & -0.162 & -0.035 & -0.340 & -0.180 & -0.134 & -0.106 & -0.194 & -0.082 & -0.328 & -0.224 \\
NFCorpus          & -0.087 &  0.004 & -0.068 & -0.011 & -0.079 & -0.057 & -0.081 & -0.018 & -0.082 & -0.061 \\
NQ                & -0.042 &  0.003 & -0.072 & -0.030 & -0.099 & -0.050 & -0.060 & -0.029 & -0.078 & -0.055 \\
HotpotQA          & -0.020 &  0.016 & -0.014 & -0.005 & -0.018 & -0.016 & -0.019 &  0.010 & -0.018 &  0.003 \\
FiQA-2018         & -0.179 &  0.006 & -0.219 & -0.142 & -0.161 & -0.123 & -0.205 & -0.127 & -0.170 & -0.127 \\
Touché-2020       & -0.168 &  0.046 & -0.226 & -0.130 & -0.178 & -0.114 & -0.175 & -0.047 & -0.247 & -0.099 \\
DBpedia           & -0.097 & -0.047 & -0.054 & -0.032 & -0.057 & -0.052 & -0.059 & -0.022 & -0.042 & -0.025 \\
SCIDOCS           & -0.040 & -0.017 & -0.058 & -0.017 & -0.048 & -0.035 & -0.073 & -0.002 & -0.054 & -0.021 \\
FEVER             & -0.200 & -0.126 & -0.037 & -0.027 & -0.043 & -0.040 & -0.010 &  0.013 & -0.029 & -0.020 \\
Climate-FEVER     & -0.314 & -0.131 & -0.153 & -0.091 & -0.091 & -0.056 & -0.105 & -0.018 & -0.083 & -0.047 \\
SciFact           & -0.025 &  0.095 & -0.049 &  0.010 & -0.041 & -0.021 & -0.048 &  0.010 & -0.058 & -0.011 \\
\midrule
\raisebox{1.5ex}{\textbf{Average}} &
  \shortstack[c]{\avgshadeinfer{-0.106}{-0.106}\\ (100\%)} &
  \shortstack[c]{\avgshadeinfer{-0.014}{-0.106}\\ (13\%)}  &
  \shortstack[c]{\avgshadeinfer{-0.104}{-0.104}\\ (100\%)} &
  \shortstack[c]{\avgshadeinfer{-0.055}{-0.104}\\ (53\%)}  &
  \shortstack[c]{\avgshadeinfer{-0.099}{-0.099}\\ (100\%)} &
  \shortstack[c]{\avgshadeinfer{-0.067}{-0.099}\\ (68\%)}  &
  \shortstack[c]{\avgshadeinfer{-0.099}{-0.099}\\ (100\%)} &
  \shortstack[c]{\avgshadeinfer{-0.039}{-0.099}\\ (39\%)}  &
  \shortstack[c]{\avgshadeinfer{-0.115}{-0.115}\\ (100\%)} &
  \shortstack[c]{\avgshadeinfer{-0.071}{-0.115}\\ (62\%)} \\
\bottomrule
\end{tabular}}
\end{minipage}
\hfill
\begin{minipage}{0.06\textwidth}
\centering
\debiascolorbarvertical[2] %
\end{minipage}
}
\end{table*}

\subsection{Comparison with CDC}

CDC~\citep{wangPerplexity2024a} estimates its correction from document-level perplexity; our method directly estimates matched human--LLM embedding displacements. Table~\ref{tab:rq3_cdc_app} compares both methods from identical original retriever scores. Across the 25 cells, our method retains $51.2\%$ of the original bias, compared with $64.5\%$ for CDC, and yields the smaller residual magnitude in 16 cells.

\begin{table*}[t]
\centering
\scriptsize
\caption{Direct comparison with CDC on five datasets and five relevance-supervised retrievers. Dataset rows report signed $\Delta$NDSR@5. The final rows report model-wise mean absolute $\Delta$NDSR@5 and the percentage retained relative to the original bias.}
\label{tab:rq3_cdc_app}
\renewcommand{\arraystretch}{1.0}
\setlength{\tabcolsep}{2.5pt}
\resizebox{\textwidth}{!}{%
\begin{tabular}{l | ccc | ccc | ccc | ccc | ccc}
\toprule
\multirow{2}{*}{Dataset}
& \multicolumn{3}{c}{ANCE}
& \multicolumn{3}{c}{coCondenser}
& \multicolumn{3}{c}{DRAGON}
& \multicolumn{3}{c}{RetroMAE}
& \multicolumn{3}{c}{TAS-B} \\
\cmidrule(lr){2-4} \cmidrule(lr){5-7} \cmidrule(lr){8-10} \cmidrule(lr){11-13} \cmidrule(lr){14-16}
& Orig. & Ours & CDC & Orig. & Ours & CDC & Orig. & Ours & CDC & Orig. & Ours & CDC & Orig. & Ours & CDC \\
\midrule
MS~MARCO
& -0.042 & -0.004 & -0.069
& -0.020 & -0.002 & -0.036
& -0.083 & -0.035 & -0.087
& -0.083 & -0.053 & -0.077
& -0.121 & -0.064 & -0.100 \\
TREC-COVID
& -0.162 & -0.035 & -0.135
& -0.340 & -0.180 & -0.274
& -0.134 & -0.106 & -0.109
& -0.194 & -0.082 & -0.118
& -0.328 & -0.224 & -0.238 \\
NQ
& -0.042 & 0.003 & -0.007
& -0.072 & -0.030 & 0.010
& -0.099 & -0.050 & 0.019
& -0.060 & -0.029 & -0.012
& -0.078 & -0.055 & 0.016 \\
FiQA-2018
& -0.179 & 0.006 & -0.127
& -0.219 & -0.142 & -0.097
& -0.161 & -0.123 & -0.105
& -0.205 & -0.127 & -0.073
& -0.170 & -0.127 & -0.135 \\
SCIDOCS
& -0.040 & -0.017 & 0.050
& -0.058 & -0.017 & -0.038
& -0.048 & -0.035 & -0.020
& -0.073 & -0.002 & -0.018
& -0.054 & -0.021 & -0.008 \\
\midrule
Mean $|\Delta|$
& 0.093 & 0.013 & 0.078
& 0.142 & 0.074 & 0.091
& 0.105 & 0.070 & 0.068
& 0.123 & 0.059 & 0.060
& 0.150 & 0.098 & 0.099 \\
Retained
& -- & 14.0\% & 83.4\%
& -- & 52.3\% & 64.2\%
& -- & 66.5\% & 64.8\%
& -- & 47.6\% & 48.5\%
& -- & 65.4\% & 66.2\% \\
\bottomrule
\end{tabular}%
}
\end{table*}

\subsection{Generalization to Qwen-Generated Passages}

We replace the Llama-2-generated passages with Qwen3-32B~\citep{qwen3technicalreport} counterparts and repeat the leave-one-dataset-out projection on the same five datasets and retrievers. Table~\ref{tab:rq3_qwen_app} shows that all five retrievers prefer Qwen-generated passages before correction in nearly every cell. Projection reduces the absolute bias in 23 of 25 cells, with model-wise retained percentages ranging from $24.3\%$ to $52.1\%$.

\begin{table*}[t]
\centering
\scriptsize
\caption{Inference-time projection on Qwen3-32B-generated passages. Dataset rows report signed $\Delta$NDSR@5. The final rows report model-wise mean absolute $\Delta$NDSR@5 and the percentage retained after projection.}
\label{tab:rq3_qwen_app}
\renewcommand{\arraystretch}{1.0}
\setlength{\tabcolsep}{3pt}
\resizebox{\textwidth}{!}{%
\begin{tabular}{l | cc | cc | cc | cc | cc}
\toprule
\multirow{2}{*}{Dataset}
& \multicolumn{2}{c}{ANCE}
& \multicolumn{2}{c}{coCondenser}
& \multicolumn{2}{c}{DRAGON}
& \multicolumn{2}{c}{RetroMAE}
& \multicolumn{2}{c}{TAS-B} \\
\cmidrule(lr){2-3} \cmidrule(lr){4-5} \cmidrule(lr){6-7} \cmidrule(lr){8-9} \cmidrule(lr){10-11}
& Original & Debias & Original & Debias & Original & Debias & Original & Debias & Original & Debias \\
\midrule
MS~MARCO
& -0.087 & -0.016
& -0.006 & 0.009
& -0.045 & -0.019
& -0.047 & -0.012
& -0.080 & -0.064 \\
TREC-COVID
& -0.117 & -0.031
& -0.135 & -0.028
& -0.164 & -0.087
& -0.250 & -0.076
& -0.106 & 0.008 \\
NQ
& -0.020 & -0.005
& -0.013 & -0.003
& -0.018 & -0.009
& -0.016 & 0.016
& -0.034 & -0.019 \\
FiQA-2018
& -0.164 & -0.037
& -0.176 & -0.115
& -0.133 & -0.063
& -0.141 & -0.046
& -0.139 & -0.100 \\
SCIDOCS
& -0.027 & -0.012
& -0.087 & -0.046
& -0.085 & -0.034
& -0.097 & 0.014
& -0.084 & -0.040 \\
\midrule
Mean $|\Delta|$
& 0.083 & 0.020
& 0.083 & 0.040
& 0.089 & 0.042
& 0.110 & 0.033
& 0.089 & 0.046 \\
Retained
& -- & 24.3\%
& -- & 48.2\%
& -- & 47.6\%
& -- & 29.8\%
& -- & 52.1\% \\
\bottomrule
\end{tabular}%
}
\end{table*}

\begin{table*}[t]
\centering
\caption{
Average passage lengths across three versions of each dataset, along with relative changes ($\Delta$) with respect to the original human-written passages. For datasets sharing the MS~MARCO corpus (e.g., DL19 and DL20),
length statistics are reported once.
}
\label{tab:length_three_versions}
\renewcommand{\arraystretch}{1.15}
\setlength{\tabcolsep}{5pt}
\begin{tabular}{lccccc}
\toprule
\textbf{Dataset} &
\textbf{Human (Orig.)} &
\textbf{Cocktail} &
\textbf{Ours} &
\textbf{$\Delta$ Cocktail (\%)} &
\textbf{$\Delta$ Ours (\%)} \\
\midrule
MS~MARCO        & 59.1  & 55.8  & 82.0  & $-5.6$ & $+38.7$ \\
TREC-COVID      & 204.8 & 171.8 & 222.0 & $-16.1$ & $+8.4$ \\
NFCorpus        & 227.6 & 184.8 & 249.3 & $-18.8$ & $+9.5$ \\
NQ              & 89.6  & 82.2  & 110.8 & $-8.2$ & $+23.6$ \\
HotpotQA        & 68.5  & 67.4  & 92.9  & $-1.6$ & $+35.6$ \\
FiQA-2018       & 135.4 & 109.3 & 147.6 & $-19.3$ & $+9.0$ \\
Touché-2020     & 167.2 & 135.5 & 172.7 & $-18.9$ & $+3.3$ \\
DBpedia         & 54.0  & 54.7  & 78.5  & $+1.3$ & $+45.3$ \\
SCIDOCS         & 174.0 & 166.0 & 182.7 & $-4.6$ & $+5.0$ \\
FEVER           & 99.4  & 91.9  & 118.7 & $-7.5$ & $+19.4$ \\
Climate-FEVER   & 87.2  & 82.0  & 107.9 & $-6.0$ & $+23.8$ \\
SciFact         & 209.4 & 178.9 & 230.2 & $-14.5$ & $+9.9$ \\
\bottomrule
\end{tabular}
\end{table*}

\section{Revisiting Length Effects in Evaluating Source Bias}
\label{app:length_effects}

\subsection{Length Bias in Neural Retrieval}

Prior work has shown that neural retrievers exhibit non-semantic bias correlated with passage length. On BEIR’s Touché-2020 argument retrieval task, \citet{thakur2024systematic} reports that neural retrievers tend to rank very short passages, many of which are non-argumentative, leading to large drops in effectiveness. Complementary evidence documents neural retrievers exhibit systematic brevity, early-position, and literal-match biases, and that these biases often result in shorter passages receiving disproportionately high scores, even outranking passages that contain the correct answer~\citep{fayyaz-etal-2025-collapse}.

This consideration is particularly relevant for evaluating \emph{source bias}, since differences between human-written and LLM-generated passages may be partially reflected in their length. The Cocktail benchmark~\citep{daiCocktail2024a}, widely used to compare LLM-generated and human-written passages, constructs LLM rewrites that are typically shorter. This leads to the central question we examine: \textit{is the observed source bias simply a manifestation of length bias?}

\subsection{Extended Rewriting for Length Control}

To answer this question, we require an evaluation setting in which passage \emph{source} can be varied while passage \emph{length} is manipulated independently of meaning. However, because the LLM-generated versions in Cocktail are generally shorter than the original human-written passages, they do not allow us to determine whether any observed preference for LLM-generated text reflects its source or simply its reduced length.

To control for length effects, we construct an additional rewritten version of each dataset that preserves the semantic content of the original passage while producing a systematically longer variant. Following the semantic-preserving rewriting protocol of previous work~\citep{daiCocktail2024a}, we use Qwen2.5-7B-Instruct to regenerate each passage based solely on its original content, using a prompt that enforces meaning preservation while generating a slightly longer paraphrase. The full prompt is provided in Appendix~\ref{app:prompt_rewrite}.

This yields three aligned versions for every dataset: the human-written version, the Cocktail version, and our length-extended version. These parallel versions enable controlled comparison in settings where rewritten passages are either shorter than or longer than their human counterparts. We report the resulting length statistics for all datasets in the next subsection.

\subsection{Length Statistics}

Table~\ref{tab:length_three_versions} reports the average passage lengths across the three versions. Consistent with prior observations, the Cocktail version reduces passage length by approximately $5$--$20\%$ across datasets. 
By contrast, our extended version increases length by $3$--$45\%$. Taken together, these variants provide controlled conditions in which LLM-generated passages can be either shorter or longer than their human-written counterparts.

The length statistics in our analysis are computed using the standard Apache Lucene tokenizer without stemming or stopword removal. As tokenization procedures differ across implementations, our computed averages show small deviations from the values reported in the Cocktail benchmark~\citep{daiCocktail2024a}. For completeness, we include those reported values in Table~\ref{tab:dataset_stat}.

\begin{table*}[t]
\centering
\caption{$\Delta$NDSR@5 results across 14 datasets for 13 neural retrievers spanning three model families.
Negative values are shaded in \textcolor{llmred}{red} to indicate a preference for LLM-generated passages, 
while positive values are shaded in \textcolor{humanblue}{blue} to indicate a preference for human-written passages. Asterisks (*) denote statistically significant deviations from zero (two-sided t-test, $p < 0.05$).}
\label{tab:rq1_main_longer_version}
\renewcommand{\arraystretch}{1.0}
\setlength{\tabcolsep}{3.5pt}
\resizebox{1.0\textwidth}{!}{
\begin{tabular}{lccccc | ccccc | ccc}
\toprule
\multirow{2}{*}{Dataset (↓)} & \multicolumn{5}{c}{Relevance-Supervised Retrievers} & \multicolumn{5}{c}{General-Purpose Embedding Models} & \multicolumn{3}{c}{Unsupervised Retrievers} \\
\cmidrule(lr){2-6} \cmidrule(lr){7-11} \cmidrule(lr){12-14}
& ANCE & TAS-B & coCondenser & RetroMAE & DRAGON & BGE & BCE & GTE & E5 & M3E & Contriever & E5-unsup & SimCSE \\
\midrule
MS~MARCO 
& \negnum{-0.026}* & \negnum{-0.042}* & \posnum{0.032}* & \negnum{-0.014}* & \negnum{-0.007}
& \posnum{0.040}* & \posnum{0.383}* & \negnum{-0.038}* & \posnum{0.157}* & \posnum{0.446}*
& \posnum{0.271}* & \negnum{-0.113}* & \posnum{0.658}* \\

DL19
& \negnum{-0.092} & \negnum{-0.039} & \negnum{-0.065}* & \negnum{-0.056} & \negnum{-0.016}
& \posnum{0.079} & \posnum{0.561}* & \negnum{-0.072} & \posnum{0.256}* & \posnum{0.554}*
& \posnum{0.320}* & \negnum{-0.079} & \posnum{0.741}* \\

DL20
& \negnum{-0.007} & \negnum{-0.033} & \negnum{-0.029} & \negnum{-0.020} & \negnum{-0.039}
& \posnum{0.212}* & \posnum{0.458}* & \negnum{-0.004} & \posnum{0.227}* & \posnum{0.648}*
& \posnum{0.367}* & \posnum{0.017} & \posnum{0.793}* \\

NQ
& \negnum{-0.050}* & \negnum{-0.045}* & \negnum{-0.035}* & \negnum{-0.002} & \negnum{-0.047}*
& \negnum{-0.085}* & \posnum{0.498}* & \posnum{0.024}* & \posnum{0.313}* & \posnum{0.320}*
& \posnum{0.197}* & \posnum{0.152}* & \posnum{0.369}* \\

NFCorpus
& \negnum{-0.071}* & \negnum{-0.031}* & \negnum{-0.025}* & \negnum{-0.055}* & \negnum{-0.018}
& \negnum{-0.031}* & \posnum{0.090}* & \negnum{-0.015} & \posnum{0.356}* & \posnum{0.049}*
& \negnum{-0.004} & \negnum{-0.165}* & \posnum{0.423}* \\

TREC-COVID
& \negnum{-0.028} & \negnum{-0.115} & \negnum{-0.056} & \negnum{-0.233}* & \negnum{-0.079}
& \negnum{-0.061} & \posnum{0.358}* & \negnum{-0.012} & \posnum{0.151}* & \posnum{0.390}*
& \negnum{-0.113}* & \posnum{0.039} & \posnum{0.569}* \\

HotpotQA
& \posnum{0.024}* & \posnum{0.067}* & \posnum{0.024}* & \posnum{0.034}* & \negnum{-0.002}
& \posnum{0.151}* & \posnum{0.354}* & \posnum{0.051}* & \posnum{0.238}* & \posnum{0.237}*
& \negnum{-0.484}* & \negnum{-0.044}* & \posnum{0.110}* \\

FiQA-2018
& \negnum{-0.195}* & \negnum{-0.095}* & \negnum{-0.065}* & \negnum{-0.154}* & \negnum{-0.103}*
& \negnum{-0.143}* & \posnum{0.607}* & \negnum{-0.016} & \negnum{-0.085}* & \posnum{0.414}*
& \posnum{0.085}* & \negnum{-0.063}* & \posnum{0.513}* \\

Touché-2020
& \posnum{0.016} & \negnum{-0.010} & \negnum{-0.050} & \negnum{-0.002} & \negnum{-0.022}
& \negnum{-0.055} & \posnum{0.667}* & \posnum{0.022} & \negnum{-0.168}* & \posnum{0.636}*
& \negnum{-0.101} & \posnum{0.003} & \posnum{0.268}* \\

DBpedia
& \negnum{-0.033} & \negnum{-0.017} & \negnum{-0.018} & \negnum{-0.012} & \negnum{-0.023}
& \posnum{0.112}* & \posnum{0.369}* & \posnum{0.046}* & \posnum{0.242}* & \posnum{0.444}*
& \negnum{-0.314}* & \posnum{0.004} & \posnum{0.371}* \\

SCIDOCS
& \negnum{-0.023}* & \negnum{-0.075}* & \negnum{-0.074}* & \negnum{-0.007} & \negnum{-0.024}*
& \negnum{-0.148}* & \posnum{0.572}* & \negnum{-0.019}* & \posnum{0.188}* & \posnum{0.453}*
& \posnum{0.159}* & \posnum{0.032}* & \posnum{0.386}* \\

FEVER
& \negnum{-0.155}* & \negnum{-0.034}* & \negnum{-0.021}* & \negnum{-0.016}* & \negnum{-0.021}*
& \posnum{0.037}* & \posnum{0.442}* & \posnum{0.061}* & \posnum{0.159}* & \posnum{0.247}*
& \negnum{-0.075}* & \posnum{0.066}* & \posnum{0.105}* \\

Climate-FEVER
& \negnum{-0.245}* & \negnum{-0.103}* & \negnum{-0.064}* & \negnum{-0.147}* & \negnum{-0.132}*
& \negnum{-0.127}* & \posnum{0.742}* & \negnum{-0.022}* & \posnum{0.487}* & \posnum{0.253}*
& \negnum{-0.027}* & \posnum{0.147}* & \posnum{0.068}* \\

SciFact
& \negnum{-0.084}* & \negnum{-0.027} & \negnum{-0.041}* & \negnum{-0.010} & \negnum{-0.011}
& \negnum{-0.016} & \negnum{-0.006} & \negnum{-0.021} & \posnum{0.133}* & \negnum{-0.059}*
& \posnum{0.173}* & \negnum{-0.071}* & \posnum{0.181}* \\

\bottomrule
\end{tabular}
}
\end{table*}

\begin{table}[t]
\centering
\caption{$\Delta$NDSR@5 results of unsupervised retrievers after MS~MARCO fine-tuning, corresponding to the same base models in Table~\ref{tab:rq1_main_longer_version}. The ``-FT'' suffix denotes fine-tuning on MS~MARCO. Negative values are shaded in \textcolor{llmred}{red} to indicate a preference for LLM-generated passages, while positive values are shaded in \textcolor{humanblue}{blue} to indicate a preference for human-written passages. Asterisks (*) denote statistically significant deviations from zero (two-sided t-test, $p < 0.05$).}
\label{tab:rq1_ctrl_longer_version}
\renewcommand{\arraystretch}{1.0}
\setlength{\tabcolsep}{4pt}
\resizebox{\columnwidth}{!}{%
\begin{tabular}{lccc}
\toprule
\multirow{2}{*}{Dataset (↓)} & \multicolumn{3}{c}{Relevance-Supervised Retrievers} \\
\cmidrule(lr){2-4}
& Contriever-FT & E5-FT & SimCSE-FT \\
\midrule
MS~MARCO      & \posnum{0.110}*  & \negnum{-0.008}* & \negnum{-0.020}* \\ 
DL19          & \negnum{-0.015}  & \negnum{-0.046}  & \negnum{-0.082} \\ 
DL20          & \posnum{0.218}*  & \posnum{0.071}   & \negnum{-0.021} \\ 
NQ            & \negnum{-0.017}* & \negnum{-0.049}* & \negnum{-0.048}* \\ 
NFCorpus      & \negnum{-0.007}  & \negnum{-0.083}* & \negnum{-0.019} \\ 
TREC-COVID    & \negnum{-0.238}* & \negnum{-0.242}* & \negnum{-0.036} \\ 
HotpotQA      & \posnum{0.111}*  & \negnum{-0.012}* & \posnum{0.035}* \\ 
FiQA-2018     & \negnum{-0.016}  & \negnum{-0.123}* & \negnum{-0.206}* \\ 
Touché-2020   & \negnum{-0.041}  & \negnum{-0.042}  & \negnum{-0.089} \\ 
DBpedia       & \posnum{0.060}   & \negnum{-0.017}  & \negnum{-0.044}* \\
SCIDOCS       & \negnum{-0.026}* & \negnum{-0.022}* & \negnum{-0.056}* \\ 
FEVER         & \posnum{0.078}*  & \negnum{-0.041}* & \posnum{0.029}* \\ 
Climate-FEVER & \negnum{-0.091}* & \negnum{-0.145}* & \negnum{-0.120}* \\ 
SciFact       & \negnum{-0.021}  & \negnum{-0.058}* & \negnum{-0.052} \\ 
\bottomrule
\end{tabular}%
}
\end{table}

\subsection{Does Length Alone Account for Retriever Preferences?}

As shown in Section~\ref{sec:rq1}, source bias is not a universal property of neural retrievers: it appears consistently only in models trained with relevance supervision, whereas unsupervised and general-purpose embedding models do not exhibit systematic source preference. A natural follow-up question is whether these supervised models seem to favor LLM-generated passages simply because the LLM passages in Cocktail are, on average, shorter than their human-written counterparts.

To assess whether length differences alone can account for the observed preferences, we repeat the source-preference evaluation on the length-controlled versions introduced above. In this setting, each LLM passage is a meaning-preserving yet systematically lengthened version of its human-written counterpart. This setup removes the original brevity advantage of LLM passages in Cocktail and, in most cases, even places them at a length disadvantage. Under these conditions, any remaining preference for LLM-generated text cannot be attributed to shorter passage length.

Table~\ref{tab:rq1_main_longer_version} shows the results for the full set of  13 retrievers. Among the relevance-supervised models, lengthening the LLM passages reduces the magnitude of the bias, often substantially, but does not eliminate it. These models continue to prefer LLM-generated passages on most datasets, despite the fact that the LLM versions are now longer than the human counterparts. A few model–dataset pairs even flip to a slight preference for human-written passages, but the dominant pattern remains a persistent bias toward LLM passages. This indicates that passage length modulates the strength of the effect but is insufficient to account for its origin.

Table~\ref{tab:rq1_ctrl_longer_version} further examines the three unsupervised retrievers from Section~\ref{sec:rq1} after MS~MARCO fine-tuning. Models that originally exhibited little or no source preference develop a consistent bias toward LLM-generated passages once trained with relevance supervision—even when the LLM passages face the same length disadvantage. In contrast, general-purpose embedding models and the unsupervised retrievers in their base form tend to favor the human passages under the length-controlled condition, reflecting their greater sensitivity to length but also their lack of a stable preference in either direction.

Taken together, these findings suggest that although passage length influences
the strength of source preference, it is insufficient to explain why the effect
persists. Lengthening the LLM passages attenuates the bias but does not remove
it, and relevance supervision still induces a preference for LLM-generated text
when length advantages are reversed. This aligns with the overall conclusion that source bias primarily arises from stylistic and distributional artifacts in supervised retrieval data.

\subsection{Prompt for Length-Controlled Rewriting}
\label{app:prompt_rewrite}
To produce the length-controlled passages used in our evaluation, we employ a prompt that enforces meaning preservation while generating slightly longer variants of the original text. The full prompt template is shown in Figure~\ref{fig:rephrasing_prompt}. While the prompt specifies a modest increase (e.g., around 10\%), LLM outputs exhibit natural variability, and some rewritten passages become substantially longer. This variability does not affect our analysis, as the experiment requires only that the LLM passages be consistently longer than the originals, not that they match a specific percentage.

\begin{figure*}[ht]
\begin{lstlisting}[frame=single, basicstyle=\ttfamily\scriptsize, xleftmargin=0pt, numbers=none, breaklines=true]
<|system|>
You are a helpful assistant.

<|user|>
Please follow the instructions below:
1. Maintain the original meaning of the input passage.
2. Make the paraphrased passage slightly longer than the original (e.g., 10% longer) while preserving the same information.
3. Output the paraphrased passage directly.

Following is the passage you need to paraphrase: 
{text}
Your answer must be formatted as:
Rewritten Text:
<your rewritten text>
\end{lstlisting}
\caption{Prompt used to generate meaning-preserving, length-extended passages with Qwen2.5-7B-Instruct.}
\label{fig:rephrasing_prompt}
\end{figure*}

\section{Additional Analyses of Supervision and Negative Sampling}
\label{app:neg-sampling}

This section examines how \emph{supervision pipelines} and \emph{negative sampling strategies} shape artifact imbalances between positives and negatives, a central component of the mechanism developed in Section~\ref{sec:rq2}.

This section begins by clarifying how different supervised retrievers construct their hard negatives, since these pipelines determine the stylistic differences between positive and negative passages. We then evaluate a retriever trained solely on NQ to examine whether the observed source bias generalizes beyond MS~MARCO. Finally, an additional control preserves BM25 mining while restricting its candidate set to passages annotated as relevant to other training queries.

\subsection{Supervision Pipelines and Hard-Negative Construction}
\label{app:neg-sampling:pipelines}

Relevance-supervised retrievers (e.g., ANCE, TAS-B, DRAGON, RetroMAE, coCondenser) typically rely on large-scale human relevance annotations, coupled with mined hard negatives.
Although their exact pipelines differ, they all follow a similar supervised contrastive setup: for each query $q$, a judged relevant passage $p^+$ is paired with one or more hard negatives $p^-$ drawn from top-ranked retrieval candidates that are either unjudged or explicitly labeled non-relevant.

Concretely, ANCE, coCondenser and DRAGON iteratively refresh hard negatives using the current dense retriever, while TAS-B and RetroMAE mainly rely on BM25- or multi-retriever–mined candidates from MS~MARCO.
In all cases, positives and negatives are sampled from \emph{different parts of the corpus}: positives concentrate on answer-like spans around annotated answers, whereas negatives come from a much broader and noisier pool.

A common pattern therefore emerges:

\begin{itemize}[noitemsep, topsep=0pt, leftmargin=1.5em]
  \item \textbf{Positives} are drawn from \emph{answer-like, high-quality passages}, often with higher fluency and information density.
  \item \textbf{Negatives} are drawn from \emph{retrieval candidates} (e.g., BM25 top-$k$ or model-mined candidates) that are more heterogeneous, noisier, and stylistically less polished.
\end{itemize}

As shown in Figure~\ref{fig:ppl_idf_supervision} and Appendix~\ref{app:linguistic_appendix}, this supervision scheme induces a systematic stylistic gap between positive and negative pools.
If we let $A$ denote a stylistic attribute (e.g., fluency), and write
\[
  \Delta A \;=\; \mathbb{E}[A^{+}] - \mathbb{E}[A^{-}],
\]
then in typical relevance-supervised training we observe $\Delta A > 0$ across multiple artifact dimensions.
This helps explain why \emph{all relevance-supervised retrievers} in Table~\ref{tab:rq1_main} exhibit pronounced source bias: supervision itself encodes a non-semantic separation between positives and negatives that overlaps with properties of LLM-generated text.

In contrast, the other model families in Table~\ref{tab:rq1_main} use supervision schemes that do \emph{not} introduce such a consistent gap:

\begin{itemize}[noitemsep, topsep=0pt, leftmargin=1.5em]
  \item \textbf{General-purpose embedding models} (e.g., E5, BGE, GTE, M3E) are trained with multi-task and weakly supervised objectives, often with symmetric sampling from the same underlying corpus.
  Positives and negatives are drawn from similar distributions (e.g., STS/NLI pairs, contrastive sentence pairs), which keeps $\Delta A$ small.
  \item \textbf{Unsupervised retrievers} (e.g., Contriever, unsupervised SimCSE) sample positives via augmentation (e.g., dropout views, adjacent sentences) and use in-batch samples as negatives.
  Both come from the same raw corpus, again implying $\Delta A \approx 0$.
\end{itemize}

These differences can therefore be unified under a single mechanism:
\emph{whether the supervised training process creates a systematic stylistic imbalance between positives and negatives}.

\begin{table}[t]
\centering
\small
\caption{Source bias of the NQ-only DPR retriever on Cocktail. We report $\Delta$NDSR@5; negative values indicate a preference for LLM-generated passages and are shaded in \textcolor{llmred}{red}. Asterisks (*) denote statistically significant deviations from zero (two-sided t-test, $p < 0.05$).}
\label{tab:dpr_nq_bias}
\renewcommand{\arraystretch}{1.0}
\setlength{\tabcolsep}{4pt}
\resizebox{\columnwidth}{!}{%
\begin{tabular}{l c @{\hspace{2em}} l c}
\toprule
Dataset & $\Delta$NDSR@5 &
Dataset & $\Delta$NDSR@5 \\
\midrule
MS~MARCO     & \negnum{-0.264}* & FiQA-2018     & \negnum{-0.631}* \\
DL19         & \negnum{-0.258}* & Touché-2020   & \negnum{-0.430}* \\
DL20         & \negnum{-0.319}* & DBpedia       & \negnum{-0.087}* \\
NQ           & \negnum{-0.106}* & SCIDOCS       & \negnum{-0.353}* \\
NFCorpus     & \negnum{-0.170}* & FEVER         & \negnum{-0.075}* \\
TREC-COVID   & \negnum{-0.503}* & Climate-FEVER & \negnum{-0.119}* \\
HotpotQA     & \negnum{-0.065}* & SciFact       & \negnum{-0.300}* \\
\bottomrule
\end{tabular}%
}
\end{table}

\subsection{Source Bias Persists Under NQ-Only Supervision}
\label{app:neg-sampling:nq}

A natural question is whether source bias is specific to MS~MARCO and its particular negative-construction scheme.
To test this, we evaluate a publicly available DPR retriever trained solely on Natural Questions (NQ)\footnote{We use the publicly released DPR encoders trained on NQ:
the question encoder (\url{https://huggingface.co/sentence-transformers/facebook-dpr-question_encoder-single-nq-base}) 
and the context encoder (\url{https://huggingface.co/sentence-transformers/facebook-dpr-ctx_encoder-single-nq-base}).}. NQ differs from MS~MARCO in several fundamental respects, including its Wikipedia-only domain, its distinct annotation and answer format, the retrieval candidates used for negative sampling, and the overall style and distribution of passages.

We evaluate this NQ-only retriever on Cocktail.
Despite the differences in corpus, annotation, and negative sampling pipeline, it still exhibits substantial source bias.
On average across the 14 datasets we observe $\Delta \mathrm{NDSR}@5 \;=\; -0.263$, and the per-dataset results are shown in Table~\ref{tab:dpr_nq_bias}.
All datasets show a clear preference for LLM-generated passages.

Although the magnitude varies across datasets, the direction of the effect mirrors that of MS~MARCO-trained retrievers.
This indicates that source bias is not an artifact of MS~MARCO, but a more general outcome of relevance supervision.
Whenever a supervised dataset introduces a stylistic discrepancy between positive and negative passages (i.e., $\Delta A > 0$), the resulting retriever tends to inherit a corresponding preference that favors LLM-generated text.

\subsection{Mining Negatives from Other Queries' Relevant Passages}
\label{app:neg-sampling:aligned-negatives}

As an additional control, we preserve hard-negative mining while changing the candidate set. For each query, BM25 searches passages annotated as relevant to other training queries and selects the highest-ranked passage that is not the current query's ground truth. The selected passage is therefore used as a negative for the current query, while retaining the answer-like properties of the annotated-relevant collection.

\begin{table}[t]
\centering
\small
\caption{
Per-dataset $\Delta$NDSR@5 and NDCG@5 when BM25 mines each query's negative from passages annotated as relevant to other training queries.
}
\label{tab:pos-pool-hardnegs-per-dataset}
\renewcommand{\arraystretch}{1.0}
\setlength{\tabcolsep}{6pt}
\begin{tabular}{lcc}
\toprule
Dataset & $\Delta$NDSR@5 & NDCG@5 \\
\midrule
MS~MARCO      &  0.027 & 0.621 \\
DL19          & -0.013 & 0.667 \\
DL20          &  0.012 & 0.661 \\
NQ            &  0.019 & 0.668 \\
NFCorpus      & -0.060 & 0.239 \\
TREC-COVID    & -0.158 & 0.574 \\
HotpotQA      &  0.013 & 0.577 \\
FiQA-2018     &  0.005 & 0.215 \\
Touché-2020   & -0.064 & 0.396 \\
DBpedia       & -0.033 & 0.430 \\
SCIDOCS       & -0.017 & 0.091 \\
FEVER         &  0.037 & 0.844 \\
Climate-FEVER & -0.069 & 0.268 \\
SciFact       & -0.023 & 0.436 \\
\midrule
Average       & -0.023 & 0.478 \\
\bottomrule
\end{tabular}
\end{table}

Table~\ref{tab:pos-pool-hardnegs-per-dataset} reports the full bias and effectiveness results for this control. Its signed average $\Delta$NDSR@5 matches in-batch-only training at $-0.023$, while average NDCG@5 remains comparable ($0.478$ versus $0.475$). Hard-negative mining therefore does not itself induce strong source preference; the decisive factor is whether the two training pools expose a systematic artifact mismatch. Together with the NQ-only results in Appendix~\ref{app:neg-sampling:nq}, this establishes negative-pool composition as a central driver of source bias.

\end{document}